\PassOptionsToPackage{table}{xcolor}
\documentclass[lettersize,journal]{IEEEtran}
\usepackage[caption=false,font=normalsize,labelfont=sf,textfont=sf]{subfig}
\usepackage{amsmath}
\usepackage{amsfonts}
\usepackage{multirow}
\usepackage{algorithm}
\usepackage{algpseudocode}
\usepackage{xcolor}

\usepackage{amsthm}
\newtheorem{remark}{Remark}
\newtheorem{theory}{Theorem}
\usepackage{bm}
\usepackage{booktabs}
\usepackage{graphicx}
\usepackage{subfig}
\usepackage{tikz}
\usepackage{pgfplots}
\usepackage[table]{xcolor}
\usepackage{url}
\usepackage{setspace}

\usepackage[normalem]{ulem}

\definecolor{fig_6_color_1}{RGB}{200,36,35}
\definecolor{fig_6_color_2}{RGB}{248,172,140}
\definecolor{fig_6_color_3}{RGB}{154,201,219}
\definecolor{fig_6_color_4}{RGB}{40,120,181}

\definecolor{fig_7_color_1}{RGB}{199,109,162}
\definecolor{fig_7_color_2}{RGB}{137,131,191}
\definecolor{fig_7_color_3}{RGB}{5,185,226}
\definecolor{fig_7_color_4}{RGB}{50,184,151}

\definecolor{fig_8_color_1}{RGB}{73,108,136}
\definecolor{fig_8_color_2}{RGB}{254,178,180}

\begin{document}

\title{Fine-Tuning Diffusion-Based Recommender Systems via Reinforcement Learning with Reward Function Optimization}

\author{Yu Hou,~\IEEEmembership{}
        Hua Li,~\IEEEmembership{}
        Ha Young Kim,~\IEEEmembership{Member,~IEEE,}
        and Won-Yong Shin,~\IEEEmembership{Senior Member,~IEEE}
\IEEEcompsocitemizethanks{\IEEEcompsocthanksitem Y. Hou and W.-Y. Shin are with the School of Mathematics and Computing (Computational Science and Engineering), Yonsei University, Seoul 03722, Republic of Korea (E-mail: \{houyu, wy.shin\}@yonsei.ac.kr).
\IEEEcompsocthanksitem H. Li is with the Department of Industrial Engineering, Yonsei University, Seoul 03722, Republic of Korea (E-mail: li\_hua611@yonsei.ac.kr).
\IEEEcompsocthanksitem Ha Young Kim is with the Graduate School of Information, Yonsei University, Seoul 03722, Republic of Korea (E-mail: hayoung.kim@yonsei.ac.kr).\\
{\it (Corresponding author: Won-Yong Shin.)}
}

}



\maketitle

\begin{abstract}
  Diffusion models recently emerged as a powerful paradigm for recommender systems, offering state-of-the-art performance by modeling the generative process of user--item interactions. However, training such models from scratch is both computationally expensive and yields diminishing returns once convergence is reached. To remedy these challenges, we propose {\bf \textsf{ReFiT}}, a new framework that integrates {\bf\textsf{Re}}inforcement learning (RL)-based \textsf{Fi}ne-\textsf{T}uning into diffusion-based recommender systems. In contrast to prior RL approaches for diffusion models depending on external reward models, \textsf{ReFiT} adopts a task-aligned design: it formulates the denoising trajectory as a Markov decision process (MDP) and incorporates a {\em collaborative signal-aware} reward function that directly reflects recommendation quality. By tightly coupling the MDP structure with this reward signal, \textsf{ReFiT} empowers the RL agent to exploit high-order connectivity for fine-grained optimization, while avoiding the noisy or uninformative feedback common in na\"ive reward designs. Leveraging policy gradient optimization, \textsf{ReFiT} maximizes {\em exact log-likelihood} of observed interactions, thereby enabling effective post hoc fine-tuning of diffusion recommenders. Comprehensive experiments on wide-ranging real-world datasets demonstrate that the proposed \textsf{ReFiT} framework (a) exhibits substantial performance gains over strong competitors (up to 36.3\% on sequential recommendation), (b) demonstrates strong efficiency with linear complexity in the number of users or items, and (c) generalizes well across multiple diffusion-based recommendation scenarios. The source code and datasets are publicly available at \url{https://anonymous.4open.science/r/ReFiT-4C60}. 
\end{abstract}


\section{Introduction}
\label{sec:introduction}

\begin{figure}[t]
\centering
\includegraphics[width=\linewidth]{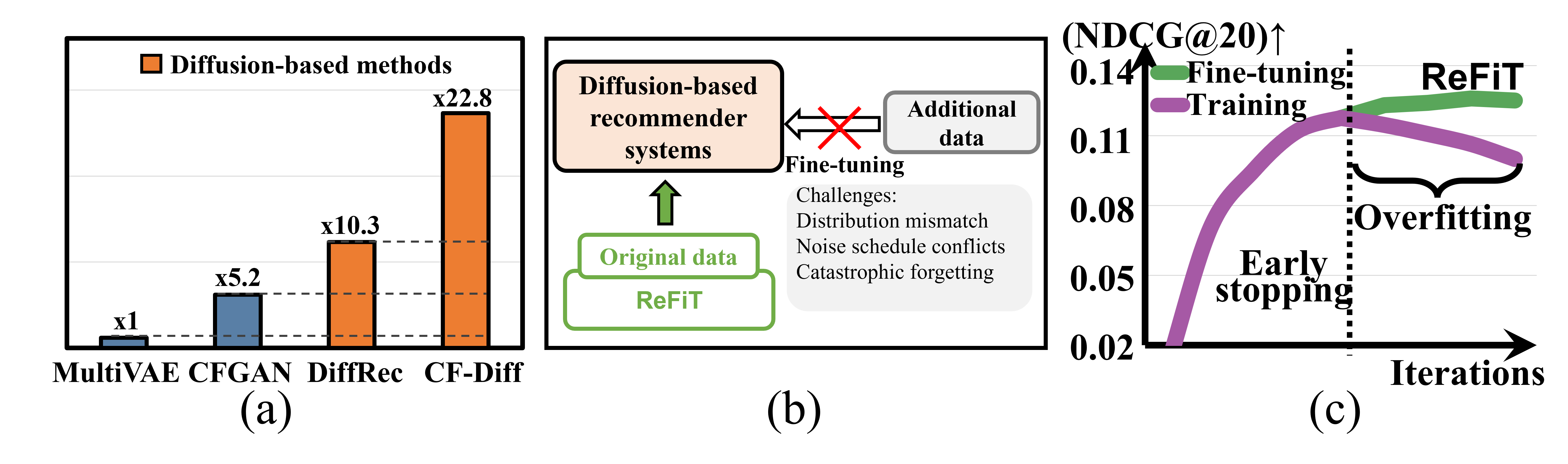} 
\caption{Examples showing (a) the runtime comparison across various generative model-based recommender systems on Anime, with each bar's height representing the relative scale of training time compared to MultiVAE, (b) the challenges associated with fine-tuning diffusion-based recommender systems using additional datasets, and (c) the recommendation accuracy in NDCG@20 over iterations for training and fine-tuning CF-Diff (a state-of-the-art diffusion-based method) on MovieLens-1M (ML-1M).}
\label{fig:motivations}
\end{figure}

\subsection{Background and Motivation}

Diffusion models \cite{sohl2015deep, ho2020denoising, cao2024survey}, a class of deep generative models, have recently demonstrated strong performance on recommendation tasks by modeling the generative process of user--item interactions. These models operate by gradually corrupting user--item interaction data through a forward-diffusion process and subsequently recovering these interactions using a neural network-based reverse-denoising process \cite{walker2022recommendation, wang2023diffusion}. Through this stochastic reconstruction process, diffusion-based recommender systems can effectively capture intricate behavioral patterns, enabling accurate prediction of unknown interactions across users. Recent advances further enhanced this capability by incorporating high-order connectivity information ({\it i.e.}, multi-hop proximity among users and items) into the diffusion modeling framework \cite{hou2024collaborative, zhu2024graph}, leading to improved utilization of collaborative signals. 

Despite these advances, several critical challenges remain. Specifically, achieving further improvements in recommendation performance using diffusion models is hindered by the following limitations:
\begin{itemize}
    \item {\bf Computational burden:} Although more sophisticated architectures could potentially improve accuracy, training diffusion-based recommenders is already substantially more computationally intensive than other generative model-based recommendation methods, such as MultiVAE \cite{liang2018variational} and CFGAN \cite{chae2018cfgan}, as shown in Fig. \ref{fig:motivations}a. Scaling up architecture complexity is thus often impractical.
    \item {\bf Data dependency in fine-tuning:} Fine-tuning on additional datasets may offer performance gains but entails non-trivial overhead in data collection, quality assurance, and task relevance evaluation. Moreover, due to the sensitivity of diffusion models to dataset-specific noise schedules, naïve fine-tuning on new datasets can lead to catastrophic forgetting (see Fig. \ref{fig:motivations}b). These make fine-tuning on the original dataset more desirable \cite{liu2022improved}.
    \item {\bf Overfitting risk:} Even when fine-tuning on the original dataset, continued training with the same loss function used in pre-training often results in overfitting beyond convergence (see Fig. \ref{fig:motivations}c). This poses a dilemma where further training does not translate into better generalization.
\end{itemize}

These challenges raise a natural question: 

``{\it How can we further enhance the performance of diffusion-based recommender systems without incurring substantial computational or data overhead?}"

\subsection{Main Contributions}

In this paper, to address this question, we establish the following two {\bf design principles (DPs)}:

\begin{itemize}
    \item {\bf DP1.} Effectively update pre-trained diffusion models for recommendation without modifying their architecture or requiring new data;
    \item {\bf DP2.} Maximize the utility of output from pre-trained models by leveraging {\it collaborative signals} during fine-tuning.
\end{itemize}

We address these principles through a reinforcement learning (RL)-aided fine-tuning approach. Specifically, we observe that sampling user--item interactions from a pre-trained diffusion model can be formulated as a multi-step decision-making process, naturally modeled as a Markov decision process (MDP). Unlike conventional fine-tuning, which continues to optimize the original loss, RL enables reward-driven optimization over sequential decisions. This motivates the use of RL---commonly employed for solving MDPs \cite{kaelbling1996reinforcement, OtterloW12, franccois2018introduction}---to guide the fine-tuning process. To this end, we develop a new framework that integrates {\bf \textsf{\underline{Re}}}inforcement learning-aided {\bf \textsf{\underline{Fi}}}ne-{\bf \textsf{\underline{T}}}uning into pre-trained diffusion-based recommender systems, named {\bf \textsf{ReFiT}}, which builds on two core ideas below.

\begin{itemize}
    \item \textbf{(\underline{Idea 1}: RL-aided fine-tuning as sequential optimization)} \textsf{ReFiT} treats each denoising step as a decision point, allowing the RL agent to adjust model behavior dynamically based on immediate feedback; this is especially valuable in recommendation tasks, as such feedback reflects how well user preferences are captured and enables more effective {\it personalization}. This approach directly optimizes a reward-guided objective---maximizing the log-likelihood of observed interactions---while avoiding overfitting and surpassing na\"\i ve fine-tuning (see Fig. \ref{fig:motivations}c). Crucially, this operates entirely on existing pre-trained diffusion models without requiring any architectural modifications or additional data. This design adheres to {\bf DP1}.
    \item {\bf (\underline{Idea 2}: Collaborative signal-aware reward design)} The cornerstone of \textsf{ReFiT}'s success lies in its reward function design. Unlike prior RL for diffusion models in computer vision  \cite{fan2023optimizing, black2023training} relying on external evaluators, \textsf{ReFiT} introduces a new reward function that incorporates {\it collaborative signals} in {\it high-order connectivities} (see Fig.~\ref{fig:high_order}). This provides richer and more reliable feedback than na\"ive reward proxies, such as binary clicks or one-step accuracy \cite{huang2021deep, zhao2018recommendations}. This design directly fulfills {\bf DP2}.
\end{itemize}

Our main contributions are summarized as follows:
\begin{itemize} 
    \item {\bf Novel methodology}: We propose \textsf{ReFiT}, the first RL-based fine-tuning framework for diffusion-based recommender systems. Central to its success is a newly designed reward function that effectively accommodates collaborative signals, allowing fine-tuning to proceed with strong guidance from pre-trained model output.
    \item {\bf Superiority in recommendation accuracy:} Extensive experiments on multiple real-world benchmark datasets show that \textsf{ReFiT} consistently surpasses state-of-the-art methods, achieving up to a {\it 36.3\%} {\it improvement} in NDCG@20 on sequential recommendation tasks. These results highlight the effectiveness of RL-aided optimization and the substantial advantage conferred by our reward modeling strategy.
    \item {\bf Computational efficiency:} \textsf{ReFiT} maintains {\it linear} computational complexity with respect to the number of users or items, which is empirically demonstrated and rigorously proven by theoretical analysis.
    \item {\bf Broad applicability:} The proposed framework generalizes across diverse diffusion-based recommendation settings, demonstrating robustness and adaptability in varying data and recommendation task scenarios.
\end{itemize}

\begin{table}[t]
\centering
\setlength\tabcolsep{2pt}
\small
  \caption{\MakeUppercase{Summary of notations.}}
  \begin{tabular}{cl}
    \toprule
    \bf Notation& \bf Description\\
    \hline
    $\mathcal{U}$ & Set of users \\
    $\mathcal{I}$ & Set of items \\
    $ {\bf u}$ & Binary vector of single user--item interactions \\
    $ {\bf u}_t$ & Single user--item interactions at $t$ time step \\
    $ T $ & Total time steps \\
    $q\left( {{\bf u}_t \left| {{\bf u}_{t - 1} } \right.} \right) $ & Forward-diffusion process in diffusion models \\
    $p_\theta  \left( {{\bf u}_{t - 1} \left| {{\bf u}_t } \right.} \right) $ & Reverse-denoising process in diffusion models\\
    $ \theta $ & Neural network in reverse-denoising process \\
    $ \beta_t$ & Noise scaling parameter at time step $t$ \\
    $\mathcal{S} $ & State space in MDPs \\
    $\mathcal{A} $ & Action space in MDPs \\
    $\mathcal{P} $ & Transition kernel in MDPs \\
    $\mathcal{R} $ & Reward function in MDPs \\
    $ {\bf s}_t $ & ${\bf s}_t \in \mathcal{S}$, the state at $t$ time step, $ {\bf s}_t = \left( {t,{\bf u}_{T - t} } \right)$\\
    $ {\bf a}_t $ & ${\bf a}_t \in \mathcal{A}$, the action at $t$ time step, ${\bf a}_t = {\bf u}_{T - t - 1}$ \\
    $P\left( {{\bf s}_{t + 1} \left| {{\bf s}_t ,{\bf a}_t } \right.} \right) $ & Probability transition to ${s}_{t+1}$ given ${\bf s}_t $ and ${\bf a}_{t}$ \\
    $R\left( {{\bf s}_t ,{\bf a}_t } \right) $ & Reward at $t$ time step \\
    $r\left( {{\bf u}_0 } \right) $ & Reward at final time step \\
    $\tau $ & Trajectory sampled from MDPs \\
    $\bar R $ & Cumulative reward from trajectory $\tau$ \\
    $ \pi_\theta$ & Policy which can be viewed as $p_\theta$ \\
    $ \mathcal{J}_{ELBO} \left( \theta  \right) $ & ELBO-based loss function\\
    $ \mathcal{J}_{RL} \left( \theta  \right) $ & RL-based loss function\\
    $ N_K$ & True positive in top-$K$ recommendations \\
    $ N_{sim - K} $ & True positive in top-$K$ of similar users  \\
    $ \alpha $ & Balancing the reward between $N_K$ and $ N_{sim - K} $\\
    \bottomrule
\end{tabular}
\label{tab:notations}
\end{table}

\subsection{Organization}

The remainder of this paper is organized as follows. Section II presents the preliminaries. Section III describes the proposed methodology, including the problem statement and technical details of our \textsf{ReFiT} framework. Comprehensive experimental evaluations are shown in Section IV. Section V presents the prior studies related to our work. Finally, we provide a summary and concluding remarks in Section VI.

Table \ref{tab:notations} summarizes the notation that is used in this paper. This notation will be formally defined in the subsequent sections when we introduce our methodology and technical details.

\section{Preliminaries}

We first provide a brief overview of diffusion-based recommender systems and RL through MDPs.

\subsection{Notations}

Let $u\in\mathcal{U}$ and $i\in\mathcal{I}$ denote a user and an item, respectively, where $\mathcal{U}$ and $\mathcal{I}$ denote the sets of all users and all items, respectively. Historical interactions of a user $u \in \mathcal{U}$ with items are represented as a binary vector ${\bf u} \in \left\{ {0,1} \right\}^{\left| \mathcal{I} \right|} $ whose $i$-th entry is 1 if there exists implicit feedback (such as a click or a view) between user $u$ and item $i \in \mathcal{I}$, and 0 otherwise.\footnote{The unbolded $u$ represents a user, while the bolded ${\bf u}$ represents a certain user's interaction vector as utilized in the proposed framework.} 


\subsection{Diffusion-Based Recommender Systems}

Previous studies have demonstrated the effectiveness of diffusion models in recommender systems \cite{walker2022recommendation, wang2023diffusion, hou2024collaborative, zhu2024graph}. Given user–item interactions for an individual user sampled from a real data distribution, ${\bf u}_0\sim p\left({{\bf u} } \right) $, the generation process can be modeled using a probabilistic diffusion framework involving two directional processes: a forward-diffusion process and a reverse-denoising process.

The forward-diffusion process is characterized as a Markovian process $q\left( {{\bf u}_t \left| {{\bf u}_{t - 1} } \right.} \right) $, where $t\in\{1,\cdots,T\}$ is the diffusion step and Gaussian noise is gradually added to the user–-item interactions ${\bf u}_0$ over $T$ time steps, producing a sequence of noisy samples $\{ {\bf u}_1 , \ldots ,{\bf u}_T \}$. At each step, noise is added based on the transition from ${\bf u}_{t-1}$ to ${\bf u}_t$ via a Gaussian distribution $q\left( {{\bf u}_t \left| {{\bf u}_{t - 1} } \right.} \right) = {\mathcal{N}}\left( {{\bf u}_t ;\sqrt {1 - \beta _t } {\bf u}_{t - 1} ,\beta _t {\bf I}} \right) $, where $\mathcal{N}$ denotes the Gaussian distribution and $\beta _t  \in \left( {0,1} \right) $ controls the noise scale at each time step $t$ \cite{sohl2015deep, ho2020denoising}, which is shown in Fig. \ref{fig:overview_sub1}.

In the reverse-denoising process, accurately estimating the distribution $q\left( {{\bf u}_{t - 1} \left| {{\bf u}_t } \right.} \right) $ is technically challenging as it requires using the entire dataset. To overcome this issue, a neural network model is employed to approximate the distribution $q({\bf u}_{t-1}|{\bf u}_t )$ \cite{ho2020denoising, wang2023diffusion}. As shown in Fig. \ref{fig:overview_sub1}, starting from ${\bf u}_T$, the reverse-denoising process progressively recovers ${\bf u}_{t-1}$ from ${\bf u}_t$ through a denoising transition step, which is modeled as $ p_\theta \!  \left( {{\bf u}_{t - 1} \! \left| {{\bf u}_t } \right.} \right)\!\! =\!\! \mathcal{N}\!\left( {{\bf u}_{t - 1} ;\!\bm{\mu} _\theta  \left( {{\bf u}_t \!, \! t} \right)\!,\!\bm{\Sigma} _\theta  {\left( {{\bf u}_t \!, \! t} \right)} } \right)$\footnote{Here, ${\bm{\mu} _\theta  \left( {{\bf u}_t, t} \right)} $ and ${\bm{\Sigma}_\theta  {\left( {{\bf u}_t, t} \right)} } $ are the mean and covariance, respectively, of the Gaussian distribution predicted by the neural network with learnable parameters $\theta $.}. The neural network can be optimized with the evidence lower bound (ELBO) using the following objective \cite{ho2020denoising, wang2023diffusion}:
\setlength{\abovedisplayskip}{1pt}
\setlength{\belowdisplayskip}{1pt}
\begin{equation}       
    \begin{array}{l}
      - \!\log p\left( {{\bf u}_0 } \right)\! \le \!\sum\limits_{t = 2}^T \!{\mathbb{E}_{q \left( {{\bf u}_t \left| {{\bf u}_0 } \right.} \right)} \!\!\left[ {\text{KL}\!\left( {q\!\left( {{\bf u}_{t - 1} \!\left| {{\bf u}_t ,\!{\bf u}_0 } \right.} \right)\!\!\left\| {p_\theta \! \left( {{\bf u}_{t - 1} \!\!\left| {{\bf u}_t } \right.} \right)} \right.} \right)} \right]}  \\ 
     \;\;\;\;\;\;\;\;\;\;\;\;\;\;\;\;\;\; - \mathbb{E}_{q\left( {{\bf u}_1 \left| {{\bf u}_0 } \right.} \right)} \left[ {\log p_\theta  \left( {{\bf u}_0 \left| {{\bf u}_1 } \right.} \right)} \right] 
     \end{array}
     \nonumber
\end{equation} 
\begin{align}
    \;\;\;\;\;\;\;\;\buildrel \Delta \over = \mathcal{J}_{ELBO} \left( \theta \right), \hspace{10em}
    \label{obj_diffusion}
\end{align}
where $\text{KL}(\cdot \| \cdot)$ denotes the Kullback–Leibler (KL) divergence between two distributions. Given a pre-trained neural network $\theta$ using the diffusion model, we can iteratively sample a trajectory $\left\{ {{\bf u}_T , \ldots ,{\bf u}_0 } \right\} $ by following the Markovian reverse-denoising process $p_\theta  \left( {{\bf u}_{t - 1} \left| {{\bf u}_t } \right.} \right) $, ultimately ending with clean (original) user--item interactions ${\bf u}_0$. 

\begin{figure}
    \centering
    \includegraphics[width=\linewidth]{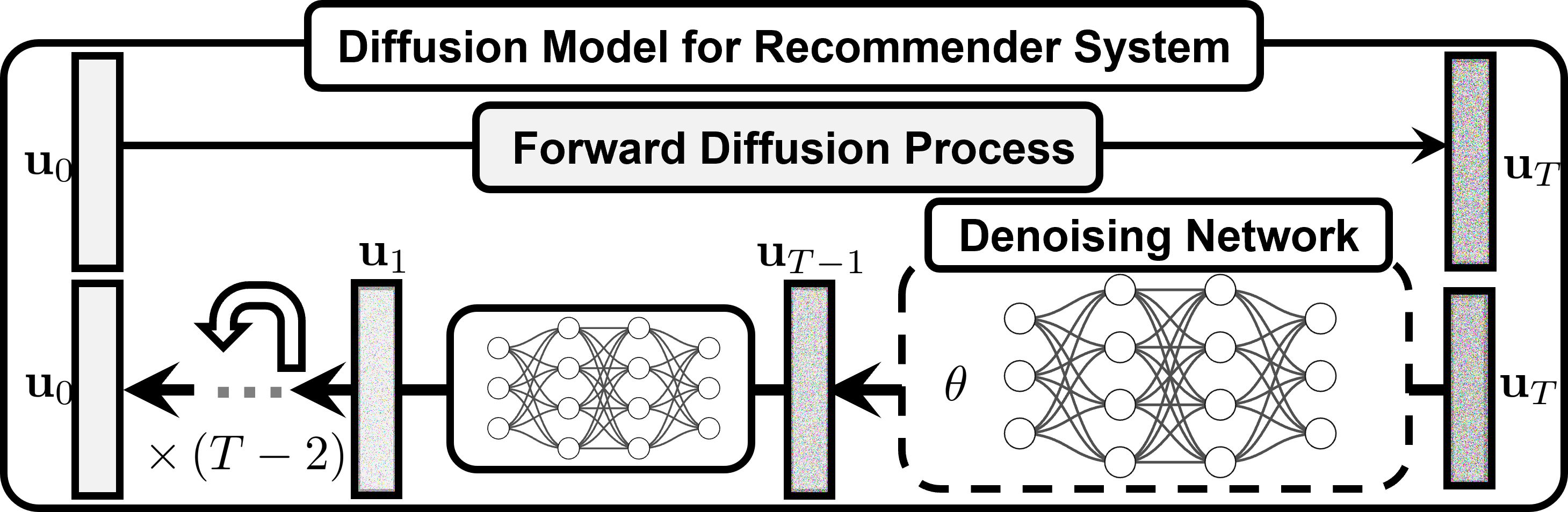} 
    \caption{Diffusion-based recommender system.}
    \label{fig:overview_sub1}
\end{figure}

\subsection{RL through MDPs}

An MDP is a formulation of sequential decision-making problems within a specific environment, defined by the tuple $\left( {\mathcal{S},\mathcal{A},\mathcal{P},\mathcal{R}} \right) $ \cite{kaelbling1996reinforcement, OtterloW12, franccois2018introduction}, where $\mathcal{S}$ denotes the state space, $\mathcal{A}$ the action space, $\mathcal{P}$ the transition kernel that specifies the probability of moving from one state to another given an action, and $\mathcal{R}$ the reward function that provides the immediate reward received after such a transition\footnote{We omit the discount factor $\gamma$ for simplicity.}.

At each time step $t$, an RL agent observes a state ${\bf s}_t  \in \mathcal{S} $, takes an action ${\bf a}_t  \in \mathcal{A} $, receives a reward $R\left( {{\bf s}_t ,{\bf a}_t } \right) $, and transitions to a new state ${\bf s}_{t + 1}  \sim P\left( {{\bf s}_{t + 1} \left| {{\bf s}_t ,{\bf a}_t } \right.} \right) $. This process is repeated as the RL agent interacts with the environment over $T$ steps, producing a sequence of states and actions known as a trajectory, denoted as $\tau  = \left( {{\bf s}_0 ,{\bf a}_0 ,{\bf s}_1 ,{\bf a}_1 , \ldots ,{\bf s}_T ,{\bf a}_T } \right) $, where ${\bf a}_T$ can be viewed as no action because it ends with ${\bf s}_T$. 

The RL agent acts according to a policy $\pi_\theta \left( {{\bf a}_t\left| {\bf s}_t \right.} \right) $, which corresponds to the probability of taking an action ${\bf a}_t$ when in a state ${\bf s}_t$ at time step $t$. The objective of the RL agent is to maximize $\mathcal{J}_{RL} \left( \theta  \right) $, which is the expected cumulative reward over trajectories sampled from its policy:
\begin{equation}
    \mathcal{J}_{RL} \left( \theta  \right) = \mathbb{E}_{\tau  \sim p\left( {\tau \left| {\pi _\theta  } \right.} \right)} \left[ {\bar R} \right],
\label{RL_objective}
\end{equation}
where $p\left( {\tau \left| {\pi _\theta  } \right.} \right) $ is the probability of obtaining a trajectory $\tau$ given the policy $\pi _\theta$ and $\bar R = \sum\nolimits_{t = 0}^{T - 1} {R\left( {{\bf{s}}_t ,{\bf{a}}_t } \right)} $ is the cumulative reward for the trajectory $\tau$. 

It is worth noting that the trainable parameters $\theta$ of the diffusion model in (\ref{obj_diffusion}) are reused in (\ref{RL_objective}). This dual usage of parameters allows for seamless integration of the RL agent into diffusion models, potentially enhancing the performance of recommender systems through the combined strengths of both approaches. Importantly, the design of the reward function serves as a critical foundation for the effectiveness of this RL framework, which will be detailed in Section \ref{sec:reward_design}.

\section{Methodology}

In this section, we elaborate on the proposed \textsf{ReFiT} framework. After stating our problem, we describe the MDP formulation that connects fine-tuning of diffusion-based recommender systems to the RL agent. We then explain how to design a new reward function to drive the learning process and how to optimize the policy based on our reward function using RL. Finally, we provide analytical findings, which theoretically validate the effectiveness of \textsf{ReFiT}. 

\begin{figure*}[t]
\centering
\includegraphics[width=0.99\linewidth]{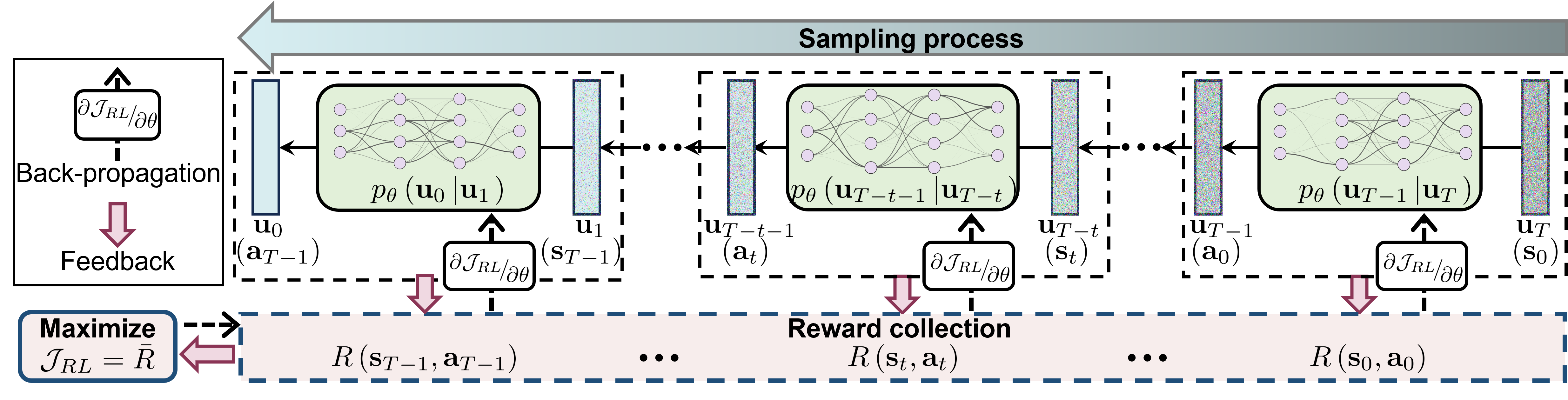} 
\captionsetup{skip=-1pt}
\caption{The schematic overview of the proposed \textsf{ReFiT} framework.}
\label{fig:overview_sub2}
\end{figure*}

\subsection{Problem Statement}

We assume the use of a diffusion model $ p_\theta$, which has already been pre-trained on a set of user--item interactions with the recommendation objective \cite{wang2023diffusion, hou2024collaborative}. According to a fixed sampling process, the denoising trajectory related to recommendations can be sampled from the diffusion model as $\tau \sim p_\theta $. In the reverse-denoising process, our objective is to fine-tune the model in the sense of maximizing the cumulative reward $\bar R$ of the trajectory, which is designed to evaluate the quality of recommendations:
    \begin{equation}
        {\hat \theta}  = \mathop {\arg \max }\limits_\theta   {\mathbb{E}_{\tau  \sim p_\theta  } \left[ {\bar R} \right]}.
    \label{our_objective}
    \end{equation}

Notably, in diffusion-based recommender systems, {\it preserving personalization} is crucial, which can be harmed through large sampling steps, implying that {\it only a few sampling steps are sufficient} \cite{wang2023diffusion, hou2024collaborative}. As a result, unlike diffusion models in other domains like computer vision, computational efficiency during sampling is not a concern in recommendation tasks. Thus, this motivates us to focus on enhancing model performance during the training or fine-tuning stages, where optimization efforts are more impactful.

\subsection{MDP Formulation}

In this subsection, we bridge between fine-tuning diffusion-based recommender systems and the RL agent through an MDP. Given the pre-trained diffusion model $p_\theta$, the sampling process can be framed as an MDP because it inherently involves a sequence of state transitions that depend only on the current state and action \cite{black2023training, fan2024reinforcement}, as depicted in Fig. \ref{fig:overview_sub2}. We regard sampling made by the diffusion model as denoising actions, and formulate the underlying MDP framework as follows:
\begin{itemize}
    \item {\bf State ($\mathcal{S}$).} In the diffusion model $p_\theta$, the state ${\bf s}_t \in \mathcal{S}$ at time step $t$ represents user--item interactions ${\bf u}_{T-t}$, which is denoted as ${\bf s}_t  = \left( {t,{\bf u}_{T - t} } \right)$. Here, the initial state ${\bf s}_0$ is the noisy user--item interactions ${\bf u}_T$ (shown on the far right in the sampling process of Fig. \ref{fig:overview_sub2}), and the final state ${\bf s}_T$ is the clean user--item interactions ${\bf u}_0$ (shown on the far left in the sampling process of Fig. \ref{fig:overview_sub2}).
    \item {\bf Action ($\mathcal{A}$).} The action ${\bf a}_t \in \mathcal{A}$ in the diffusion model is the decisions to denoise user--item interactions, sampled from the distribution $p_\theta$ based on the state ${\bf s}_t$. The action ${\bf a}_t$ is denoted as ${\bf a}_t  = {\bf u}_{T - t - 1}$. Here, the first action ${\bf a}_0$ is ${\bf u}_{T-1}$, which can be viewed as the behavior of sampling from $p_\theta$ given the initial state ${\bf s}_0$, and the final action leads to the ultimate clean user--item interactions ${\bf u}_0$, as shown in Fig. \ref{fig:overview_sub2}.
    \item {\bf State transition ($\mathcal{P}$).} The transition probability is defined as the probability of transitioning from ${\bf s}_t$ to ${\bf s}_{t+1}$ given ${\bf a}_t$:
    \begin{equation}
        P\left( {{\bf s}_{t + 1} \left| {{\bf s}_t ,{\bf a}_t } \right.} \right) = \delta \left( {{\bf s}_{t + 1}  - f\left( {{\bf s}_t ,{\bf a}_t } \right)} \right),
    \end{equation}
    where $\delta \left(  x \right) $ represents the Dirac delta distribution with nonzero density only at $x=0$ and $f\left( {{\bf s}_t ,{\bf a}_t } \right) \! = \!\left( {t\! +\! 1,{\bf u}_{T\! -\! (t + 1)} } \right) $.
    \item {\bf Reward ($\mathcal{R}$).} At each state ${\bf s}_t$, after taking the action ${\bf a}_t$, the reward $R\left( {{\bf s}_t ,{\bf a}_t } \right)$ is assigned as follows:
    \begin{equation}
        R\left( {{\bf s}_t ,{\bf a}_t } \right) = \left\{ \begin{array}{l}
         r\left( {{\bf u}_0 } \right)\;\;\;\;if\;t = T-1 \\ 
         0\;\;\;\;\;\;\;\;\;\;\;\;\text{otherwise}. \\ 
         \end{array} \right.
    \end{equation}
    We collect the reward for all time steps to compute the cumulative reward, as illustrated in the bottom block of Fig. \ref{fig:overview_sub2}. Note that only the termination state will receive a positive reward; this strategy is widely used in the design of reward functions \cite{black2023training, fan2024reinforcement}. This MDP is repeated as the RL agent interacts with the diffusion model over $T$ time steps, producing a sequence of states and actions known as a trajectory. The cumulative reward of this trajectory is given by $\bar R = \sum\nolimits_{t = 0}^{T - 1} {R\left( {{\bf{s}}_t ,{\bf{a}}_t } \right)}  = r\left( {{\bf{u}}_0 } \right)$.
\end{itemize}

\vspace{-0.8em}

\subsection{Reward Function Design}
\label{sec:reward_design}
In RL, the reward function provides valuable feedback to the RL agent ({\it i.e.}, the neural network within the diffusion model), helping the agent understand the consequences of its actions. A well-designed reward function ensures that the RL agent receives clear signals about what behaviors are desired, thereby guiding the learning process efficaciously. In our study, inspired by the fact that capturing collaborative signals plays a pivotal role in developing collaborative filtering (CF) techniques~\cite{wang2019neural, hou2024collaborative}, we design a new reward function that judiciously integrates collaborative signals in high-order connectivities into the \textsf{ReFiT} framework to enhance personalized recommendations. As shown in Fig. \ref{fig:high_order}, by leveraging shared preferences among behaviorally similar users, the reward function reinforces the relevance of recommendations to better reflect individual interests. Specifically, the reward function $r({\bf u}_0)$ for a given user provides feedback based on the top-$K$ recommendation accuracy, representing the true positive within top-$K$ recommendations, of not only the target user but also his/her similar users discovered by multi-hop relationships. We formally characterize our reward function using the top-$K$ \underline{r}ecommendation \underline{a}ccuracy with \underline{c}ollaborative \underline{s}ignals, dubbed the {\it RACS reward function}, as follows:
\begin{equation}
    r_\text{RACS}\left( {{\bf u}_0 } \right) = \alpha N_K  + (1-\alpha) N_{sim - K}, 
    \label{eq:our_reward}
\end{equation} 
where $N_K$ represents the number of true positive recommended items in the top-$K$ recommendations of the target user; $N_{sim - K}$ counts true positive recommended items in the top-$K$ recommendations of similar users; and the hyperparameter $\alpha$ balances between the two terms $N_K$ and $N_{sim-K}$.\footnote{The number of true positive (correct) recommended items can be counted given training data of user--item interactions.} Here, similar users to a target user can be identified by calculating the cosine similarity between the target user and all other users based on user--item interactions. The effectiveness of our collaborative signal-aware reward function $r_\text{RACS}({\bf u}_0)$, compared with other reward designs, will be empirically validated in Section \ref{sec:RQ2}. By capturing both direct and high-order connectivities for each user, the reward function provides richer personalized feedback, thereby further enhancing personalized recommendations.
\begin{figure}
    \centering
    \includegraphics[width=0.9\linewidth]{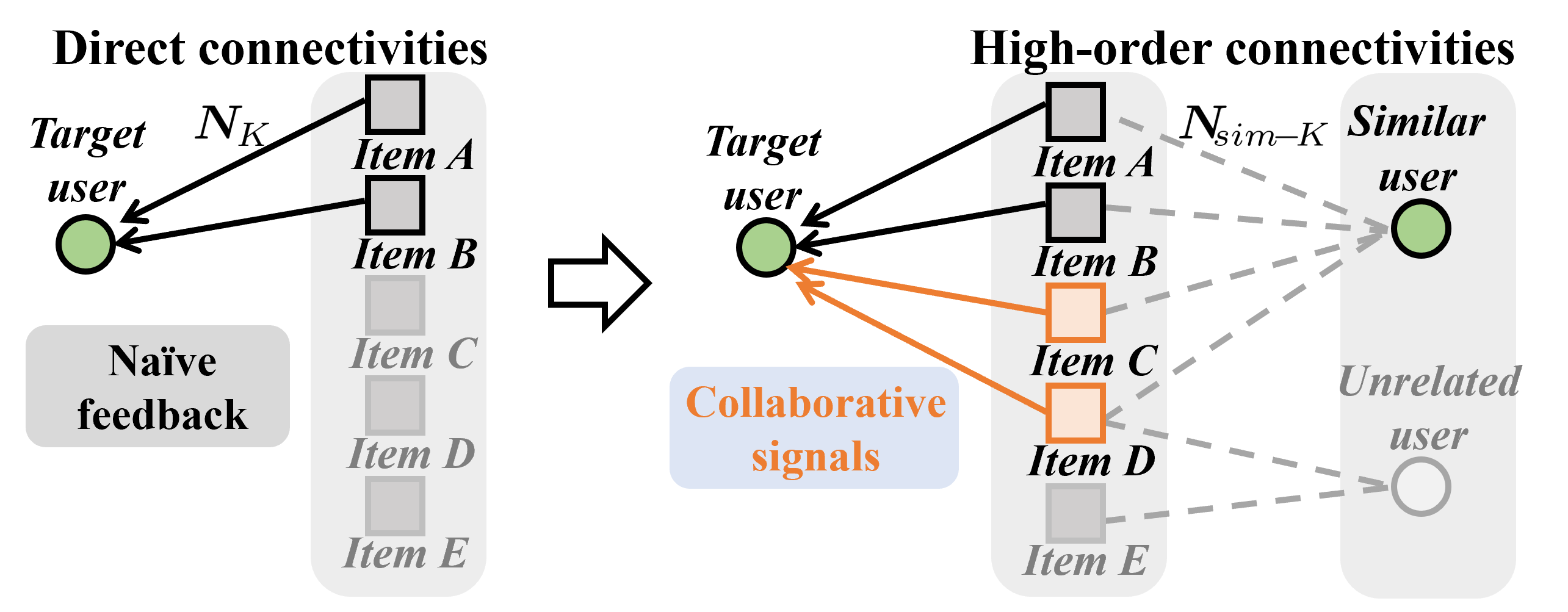}
    \caption{Enhancing recommendation quality through high-order collaborative signals. By leveraging high-order connectivity, the model can infer additional relevant items for the target user by utilizing indirect collaborative signals from similar users---signals that are often overlooked when relying solely on na\"{\i}ve user feedback.}
    \label{fig:high_order} 
\end{figure} 
\begin{remark}
It is worth noting that the designed reward function, which captures collaborative signals, is highly effective and broadly applicable. As long as user–item interactions are available, this reward can be seamlessly integrated into various diffusion-based recommendation tasks, making it a versatile and powerful design component.
\end{remark}

\subsection{Policy Optimization}
\label{sec:opt}
In this subsection, we aim to optimize the trainable parameters $\theta $ in the diffusion model $p_\theta $ in the sense of maximizing the cumulative reward, as outlined in the RL-based objective function in (\ref{our_objective}). By formulating the diffusion model as an MDP, we can treat the policy $\pi_\theta $ in (\ref{RL_objective}) as analogous to the diffusion model $p_\theta$ in (\ref{our_objective}). Specifically, we fine-tune the parameters $\theta$, initialized with a pre-trained diffusion model $p_{\theta_0}$, to maximize the objective function in (\ref{our_objective}). This approach allows us to leverage the RL agent to enhance the performance of diffusion-based recommender systems.

{
\begin{algorithm}[t]
\caption{\textsf{ReFiT}}\label{alg:cap}
\begin{algorithmic}[1]
\Require Pre-trained model $p_\theta$, where $\theta = \theta_0$, all users $\mathcal{U}$, reward $r\left( \cdot \right)$, number of iterations $Iters$, time step $T$, learning rate $l$. 
\While{$i < Iters$}
\State Sample a batch of users' interactions $U \subset \mathcal{U} $
\For{all user in $U $}
\State Sampling ${\bf u}_{T:0}  \sim p_\theta $
\State Compute cumulative reward $ \bar R \!=\! r\left( {{\bf u}_0 } \right)$ in (\ref{eq:our_reward})

\State Compute gradient $ \nabla _\theta = \nabla \mathcal{J}_{RL} \left( {\theta  } \right)$ in (\ref{our_gradient})
\State $\theta  \leftarrow \theta  + l \cdot \nabla _\theta  $

\EndFor
\State $i = i + 1 $
\EndWhile
\end{algorithmic}
\end{algorithm}
}

The overall procedure of the proposed \textsf{ReFiT} framework is detailed in Algorithm \ref{alg:cap}. For fine-tuning with an RL agent, we iteratively update $\theta$. In each iteration, we first sample a batch of users' interactions (refer to line 2). We then collect a set of denoising trajectories, each denoted as ${\bf u}_{T:0}  = \left( {{\bf u}_T ,{\bf u}_{T-1} , \ldots ,{\bf u}_1 ,{\bf u}_0 } \right) $, by sampling from the given diffusion model (refer to line 4). Additionally, we compute the corresponding cumulative reward $\bar R = r\left( {{\bf u}_0 } \right)$ for each trajectory (refer to line 5). To maximize $\mathcal{J}_{RL}(\theta)$ in (\ref{RL_objective}), we use the REINFORCE algorithm \cite{williams1992simple} with gradient ascent, a simple yet effective policy gradient method for this task, although other policy gradient methods can also be applied (see Section \ref{sec:com_RL_fine_tuning} for investigating its effectiveness). The gradient of $\mathcal{J}_{RL}(\theta)$ is computed as follows (refer to line 6):
\begin{equation}
    \nabla \mathcal{J}_{RL} \left( {\theta  } \right) = \mathbb{E}_{{\bf u}_{T:0}  \sim p_\theta} \left[ {\sum\limits_{t = 1}^T {\nabla _\theta  \log p _\theta  \left( {{\bf u}_{t-1} \left| {{\bf u}_t } \right.} \right){r\left( {{\bf u}_0 } \right)} } } \right],
    \label{our_gradient}
\end{equation}
where the diffusion model can be viewed as the policy $p_\theta  \left( {{\bf u}_{t - 1} \left| {{\bf u}_t } \right.} \right) = \pi _\theta  \left( {{\bf a}_t \left| {{\bf s}_t } \right.} \right) $ to decide the next action. Here, the expectation is estimated by taking the average over all collected denoising trajectories using the Monte Carlo sampling method \cite{mohamed2020monte}. The inference details of \textsf{ReFiT} are provided in AppendixA.II.


\subsection{Theoretical Analyses}
\label{sec:theorem}

First, we seek to formally establish a connection between the loss of \textsf{ReFiT} and the ELBO-based loss. Unlike the case of optimizing diffusion models alongside the ELBO-based loss \cite{ho2020denoising, wang2023diffusion}, we are capable of optimizing the {\it exact log-likelihood} of the user--item interactions ($ {\nabla _\theta  \log p _\theta  \left( {{\bf u}_{t-1} \left| {{\bf u}_t } \right.} \right)} $) at each denoising step, instead of approximating the log-likelihood induced by the ELBO in (\ref{obj_diffusion}). This enables {\it direct} optimization of the true log-likelihood, thereby allowing further performance improvements even after the gains from pre-training have saturated. In other words, our formulation avoids the suboptimality of standard ELBO-based training. Formally, let $p_\theta$ denote the pre-trained model and $r\left( {{\bf u}_0 } \right) $ the reward function. During fine-tuning of $p_\theta$, we have
\begin{equation}
    \begin{array}{l}
     \mathbb{E}_{p_{\theta }} \left[ { - r\left( {{\bf u}_0 } \right)\log p_\theta  \left( {{\bf u}_0 } \right)} \right] \\ 
    \le \!\! \mathbb{E}_{p_{\theta}} \!\! \left[ {\! r \! \left( {\!{\bf u}_0 }\! \right) \!\! \sum\limits_{t = 2}^T \!\!{\mathbb{E}_{q\left( {\!{\bf u}_t \!\left| {{\bf u}_0 \!\! } \right.} \right)} \!\! \left[ {\textup{KL} \! \left( {\!q\left( {\!{\bf u}_{t - 1} \! \left| {{\bf u}_t ,{\bf u}_0 \! } \right.} \right) \! \left\| {p_\theta \!\! \left( {\!{\bf u}_{t - 1} \!\left| {{\bf u}_t \!} \right.} \!\right)} \right.} \right)} \right]} } \right] \! + \! C, \\ 
     \end{array}
     \label{theory1}
\end{equation}
where $\mathbb{E}_{p_\theta  } \left[  \cdot  \right] $ indicates the expectation over all trajectories sampled from $p_\theta$; ${\mathbb{E}_{q\left( {{\bf u}_t \left| {{\bf u}_0 } \right.} \right)} \left[ {\textup{KL} \left( {\cdot\left\| { \cdot } \right.} \right)} \right]} $ corresponds to the per-step ELBO loss in (\ref{obj_diffusion}); and $C$ is a constant. Then, we would like to provide the following theoretical insight.

\begin{remark}
    The left-hand side of (\ref{theory1}), $ \mathbb{E}_{p_\theta}\!\! \left[ {\!-\!\log \! p _\theta \!  \left( {{\bf u}_0} \right){\! r \!\left( {{\bf u}_0 \!} \right)\!} }  \right]$, serves as our loss in \textsf{ReFiT}.\footnote{This is typically used as a surrogate loss function in practice when the REINFORCE algorithm is employed for optimization \cite{williams1992simple, black2023training}, instead of the objective function in (\ref{RL_objective}). This loss is derived by tracing back from the gradient in (\ref{our_gradient}), where $p_\theta  \left( {{\bf u}_0 } \right) = \prod\limits_{t = 1}^T {p_\theta  \left( {{\bf u}_{t - 1} \left| {{\bf u}_t } \right.} \right)} $. Moreover, a negative sign is added as we aim to minimize the loss function, whereas the original objective is the cumulative reward maximization.} More specifically, the term $\log \! p _\theta \! \left( {{\bf u}_0} \right){\!r\!\left( {{\bf u}_0 } \right)}$ represents a reward-weighted log-likelihood, which biases the learning process towards actions that yield higher rewards. The right-hand side of (\ref{theory1}) is the loss function derived from (\ref{obj_diffusion}) and can be viewed as the reward-weighted ELBO-based loss \cite{fan2024reinforcement}. Both terms leading to higher rewards have a higher impact on updating the policy, reinforcing the agent's preference for actions that result in more accurate recommendations. However, continuing to fine-tune with such a reward-weighted ELBO-based loss is unlikely to surpass the performance achieved by the pre-trained model with the same loss function due to overfitting. Therefore, optimizing the exact log-likelihood ({\it i.e.}, $ \mathbb{E}_{p_\theta} \left[ {-\log p _\theta  \left( {{\bf u}_0} \right){r\left( {{\bf u}_0 } \right)} }  \right]$) can lead to further performance improvements, even after the pre-trained diffusion model has converged.
\end{remark}

We refer to Appendix B for further details on the related theoretical discussion. Second, to validate the scalability of our \textsf{ReFiT} method, we analytically show its computational complexity by establishing the following theorem.

\begin{theory}
\label{sec:theorem2}
The computational complexity of \textsf{ReFiT} is given by $ \mathcal{O}\left( {\max \left\{ {\left| \mathcal{U} \right|,\left| \mathcal{I} \right|} \right\}} \right) $.
\end{theory}

\begin{proof}
    We begin by breaking down the proof into three main stages: 1) sampling users, 2) generating recommended items, and 3) computing the reward function. First, the computational complexity of sampling users is $\mathcal{O}\left( {\left| U_{sample} \right|} \right) $, where $\left| U_{sample} \right|$ is the total number of sampled users. Second, for each user, generating recommended items requires calculating the relevance of each item, yielding a computational complexity of $\mathcal{O}\left( {\left| \mathcal{I} \right|} \right) $. Third, the reward computation consists of two steps for each user, including (i) discovering similar users and (ii) counting the true positives, which correspond to $N_K$ and $N_{sim-K}$, respectively. Finding the top-$d$ similar users leads to a computational complexity of $O\left( {\left| \mathcal{U} \right| \log {d}} \right) $, which can be handled in the preprocessing stage. Computing $N_K$ requires $\mathcal{O}\left( {K} \right) $ operations, as it involves comparing $K$ recommendations with the ground truth. Computing $N_{sim-K}$ has a complexity of $\mathcal{O}\left( {d K} \right) $, as it compares $K$ recommendations with the ground truth among $d$ similar users. Therefore, the total complexity of \textsf{ReFiT} is bounded by $ \mathcal{O}\left( {\left| {U_{sample} } \right| \left( {\left| \mathcal{U} \right| {\log {d}}  + \left| \mathcal{I} \right| + K + d K} \right)} \right) $. Due to the fact that $d$, $K$, and $\left| {U_{sample} } \right|$ are constants independent of dataset scaling, the total computational complexity is simplified to $ \mathcal{O}\left( {\max \left\{ {\left| \mathcal{U} \right|,\left| \mathcal{I} \right|} \right\}} \right) $. This completes the proof of Theorem \ref{sec:theorem2}.
\end{proof}

From Theorem \ref{sec:theorem2}, one can see that the computational complexity of \textsf{ReFiT} scales {\it linearly} with respect to the number of users or items. This is empirically validated in Section IV-B6.



\section{Experimental Evaluation}

In this section, we systematically conduct extensive experiments to answer the following six key research questions (RQs):

\begin{itemize}
    \item {\bf RQ1}: How much does \textsf{ReFiT} improve the recommendation accuracy over benchmark recommendation methods for standard CF?
    \item {\bf RQ2}: How does our reward function in \textsf{ReFiT} contribute to the recommendation accuracy?
    \item {\bf RQ3}: How much is \textsf{ReFiT} effective in comparison with other fine-tuning strategies?
    \item {\bf RQ4}: How does \textsf{ReFiT} behave on other downstream recommendation tasks?
    \item {\bf RQ5}: How does the key parameter affect the performance of \textsf{ReFiT}?
    \item {\bf RQ6} How does \textsf{ReFiT} perform in terms of memory usage and computational efficiency?
\end{itemize}

We basically carry out experiments for {\it four} types of widely-used recommendation tasks in which diffusion models were developed in the literature: recommendation for {\it standard} CF \cite{zhang2019deep, bobadilla2013recommender} and {\it sequential} recommendation \cite{ijcai/WangHWCSO19, fang2020deep}, {\it social} recommendation \cite{fan2019graph}, and {\it point-of-interest (POI)} recommendation \cite{zhao2020go}. Nevertheless, we focus primarily on showcasing a full set of experimental results for the recommendation task for standard CF; we refer to Section \ref{exp:rq5} for experimental results on the other three recommendation tasks. 
\subsection{Experimental Settings}

\begin{table}[t]
\small
  \captionsetup{skip=0pt}
  \caption{The statistics of the datasets used for four downstream recommendation tasks, including standard CF, sequential recommendation, social recommendation, and POI recommendation. Here, ``Int.'' = interactions, ``Seq.'' = average sequence length, and ``Links'' = social connections.}
  \resizebox{0.48\textwidth}{!}{
  \begin{tabular}{cccccc}
    \toprule
    {\bf Task} & {\bf Dataset} & {\bf \#User} & {\bf \#Item} & {\bf \#Int.} & {\bf Extra Stat}\\
    \midrule
    \multirow{3}*{{\bf Standard CF}} &{\bf MovieLens-1M} & 5,949 & 2,810 & 571,531 & --\\
    &{\bf Yelp2018} & 31,668 & 38,048 & 1,561,406 & -- \\
    &{\bf Anime} & 73,515 & 11,200 & 7,813,737 & --\\
    \midrule
    \multirow{3}*{\shortstack{{\bf Sequential} \\ {\bf Recommendation}}} & {\bf YooChoose} & 128,468& 9,514& 539,436 & Seq.: 4.20\\
    & {\bf KuaiRec} & 92,090& 7,261& 737,163 & Seq.: 8.00\\
    & {\bf Zhihu} & 11,714& 4,838 & 77,712& Seq.: 6.63 \\
    \midrule
    \multirow{3}*{\shortstack{{\bf Social}\\{\bf Recommendation}}} & {\bf Ciao} & 1,925& 15,053&23,223 & \#Links: 65.08k\\
    & {\bf Yelp-sub} & 99,262& 105,142 & 672,513 & \#Links: 1.29m\\ 
    & {\bf Epinions} & 14,680& 233,261& 447,312 & \#Links: 632.14k\\
    \midrule
    \multirow{3}*{\shortstack{{\bf POI} \\ {\bf Recommendation}}} & {\bf Foursquare} & 2,321& 5,596&194,108 & --\\
    & {\bf TKY}  & 2,293& 15,177& 494,807 & --\\
    & {\bf NYC} & 1,083& 9,989& 179,468 & --\\
    \bottomrule
\end{tabular}
}
\label{table:datasets}
\vspace{-1.0em}
\end{table}

\begin{table*}[!t]\centering
\setlength\tabcolsep{6.0pt}
\small
  \captionsetup{skip=0.0pt}
  \caption{Performance comparison among \textsf{ReFiT} and recommendation competitors for the three benchmark datasets. Here, the best and second-best performers are highlighted by bold and underline, respectively. For the ML-1M, Yelp2018, and Anime datasets, a paired \textit{t}-test is conducted, yielding p-values of 0.0209, 0.0328, and 0.0134, respectively, all below the threshold of 0.05, indicating statistically significant results.}
  \label{tab:comparison}
  \begin{tabular}{c|cccc|cccc|cccc}
    \toprule[1pt]
    \multicolumn{1}{c|}{}&\multicolumn{4}{|c|}{{\bf ML-1M}}&\multicolumn{4}{c|}{{\bf Yelp2018}}&\multicolumn{4}{c}{{\bf Anime}}\\
    \cmidrule{1-13}
           {\bf Method} & {\bf R@10} & {\bf R@20}& {\bf N@10}& {\bf N@20}& {\bf R@10}& {\bf R@20}& {\bf N@10}& {\bf N@20}& {\bf R@10}& {\bf R@20}& {\bf N@10}& {\bf N@20}\\
    \midrule[1pt]
    {\bf NICF}& 0.0682& 0.1170& 0.0823& 0.0762& 0.0034& 0.0038& 0.0032& 0.0041 & 0.0531 & 0.0774 & 0.0716 & 0.0591\\
    {\bf FCPO}& 0.0449& 0.0803& 0.0520 &0.0439 & 0.0022& 0.0024& 0.0027& 0.0033 & 0.0472& 0.0517 & 0.0658 & 0.0552\\
    {\bf NGCF} & 0.0864& 0.1484& 0.0805&0.1008& 0.0275& 0.0482& 0.0313& 0.0391 & 0.1924& 0.2888 & 0.3515 & 0.3485\\
    {\bf LightGCN}& 0.0824& 0.1419& 0.0793& 0.0982& 0.0328& 0.0566& 0.0375& 0.0462 & 0.2071 & 0.3043 & 0.3937 & 0.3824\\
    {\bf SGL} & 0.0806& 0.1355& 0.0799 & 0.0968 & 0.0339 & 0.0595 & 0.0403 & 0.0497 & 0.1994 &0.2918 & 0.3748 & 0.3652 \\ 
    {\bf CFGAN}& 0.0684& 0.1181& 0.0663& 0.0828& 0.0163& 0.0278 & 0.0187 &0.0233 & 0.1664 & 0.2551& 0.3675& 0.3546\\    
    {\bf MultiDAE}& 0.0769& 0.1335& 0.0737& 0.0919& 0.0348& 0.0567& 0.0405& 0.0497 & 0.2142& 0.3085 & 0.4177 & 0.4125\\
    {\bf RecVAE}& 0.0835& 0.1422& 0.0769& 0.0963& 0.0344& 0.0587& 0.0393& 0.0482 & 0.2137 & 0.3068 & 0.4105 & 0.4068\\
    {\bf HDRM} & 0.1071 & 0.1834 & 0.0914 & 0.1168&  0.0337 & 0.0591 & 0.0404  & 0.0491 &  0.2148 & 0.3124 & 0.5133 & 0.4793\\
    \midrule[1pt]
    {\bf DiffRec}& 0.1058& 0.1781& 0.0901& 0.1131& 0.0351& 0.0597& 0.0414& 0.0499 & 0.2193 & 0.3249 & 0.5196 & 0.4845\\
    \rowcolor{lightgray!30}
    {\bf \textsf{ReFiT(DiffRec)}}& \underline{0.1083}& 0.1799& \underline{ 0.0918}& 0.1161& 0.0355& 0.0602& 0.0417& 0.0504 &0.2231 & \underline{0.3266} & 0.5211 & 0.4861\\
    \cmidrule{1-13}
    {\bf CF-Diff}& 0.1077 & \underline{0.1843}& 0.0912& \underline{0.1176}& \underline{0.0363}& \underline{0.0608} & \underline{0.0425}& \underline{0.0509} &\underline{0.2263} & 0.3265 & \underline{0.5271} & \underline{0.4873}\\
    \rowcolor{lightgray!30}
    {\bf \textsf{ReFiT(CF-Diff)}}& {\bf 0.1103}& {\bf 0.1866}& {\bf 0.0927}& {\bf 0.1185}& {\bf 0.0367}& {\bf 0.0618}& {\bf 0.0428}& {\bf 0.0516} &{\bf 0.2283} & {\bf 0.3303} & {\bf 0.5319} & {\bf 0.4921}\\
    \bottomrule[1pt]
  \end{tabular}
  \vspace{-1.0em}
\end{table*}

{\bf Datasets.} We conduct our experiments on three real-world datasets widely adopted for evaluating the performance of recommendations for standard CF, which include ML-1M\footnote{\url{https://grouplens.org/datasets/movielens/1m/}.} and two larger datasets, Yelp2018\footnote{\url{https://www.yelp.com/dataset/}.} and Anime\footnote{https://www.kaggle.com/datasets/CooperUnion/anime-recommendations-database.}. In addition, we use three datasets for sequential recommendation (YooChoose, KuaiRec, and Zhihu\footnote{\url{https://anonymous.4open.science/r/ReFiT_DreamRec-816D}.}), three datasets for social recommendation (Ciao, Yelp-sub, and Epinions\footnote{\url{https://anonymous.4open.science/r/ReFiT_RecDiff-38D1}.}), and three datasets for POI recommendation (Foursquare, TKY, and NYC\footnote{\url{https://anonymous.4open.science/r/ReFiT_Diff-POI-F37E}.}). Table \ref{table:datasets} summarizes the statistics of each dataset. 

\noindent\textbf{Competitors.} To comprehensively demonstrate the superiority of \textsf{ReFiT}, we present eleven benchmark recommendation methods for standard CF, including two RL-based CF methods (NICF \cite{zou2020neural}, FCPO \cite{ge2021towards}), three graph convolution-based CF methods (NGCF \cite{wang2019neural}, LightGCN \cite{he2020lightgcn}, SGL \cite{wu2021self}), three generative-based CF methods (CFGAN \cite{chae2018cfgan}, MultiDAE \cite{liang2018variational}, RecVAE \cite{shenbin2020recvae}), 
 and three diffusion-based CF methods (DiffRec \cite{wang2023diffusion}, CF-Diff \cite{hou2024collaborative}, and HDRM \cite{yuan2025hyperbolic})\footnote{We include HDRM \cite{yuan2025hyperbolic}, a recent diffusion-based CF method, as one of benchmark methods, but exclude developing \textsf{ReFiT} fine-tuned on HDRM. This is because, unlike DiffRec and CF-Diff, which directly generate user--item interaction sequences via diffusion processes, HDRM leverages diffusion models to produce user and item embeddings and then computes their similarity, which does not align with our design objectives.}. Additionally, we use DreamRec \cite{yang2023generate}, RecDiff \cite{li2024recdiff}, and Diff-POI \cite{qin2023diffusion} as benchmark methods for sequential recommendations, social recommendations, and POI recommendations, respectively. We refer to Appendix C.II for details of these competing methods.

\noindent\textbf{Performance metrics.} We follow the full-ranking protocol \cite{he2020lightgcn} by ranking all the non-interacted items for each user. In our study, we adopt two widely used ranking metrics, Recall@$N$ (R@$N$) and NDCG@$N$ (N@$N$), where $N \in \left\{ {10,20} \right\} $.

\noindent\textbf{Implementation details.} \textsf{ReFiT} only requires a pre-trained diffusion-based recommender system but {\it no extra/new} datasets. We use the pre-trained DiffRec \cite{wang2023diffusion} and CF-Diff \cite{hou2024collaborative} models when available; otherwise, we pre-train them using the original settings. The same data split as the pre-training stage is used in \textsf{ReFiT}. We use the best hyperparameters of competitors and \textsf{ReFiT} obtained by extensive hyperparameter tuning on the validation set. We use the Adam optimizer \cite{kingma2014adam}, where the batch size is selected in the range of $\left\{ {32,64,128} \right\} $. The hyperparameters used in the pre-trained diffusion model ({\it e.g.,} the noise schedule $\beta_t$ and the diffusion step $T$) are fixed and essentially follow the settings in \cite{wang2023diffusion, hou2024collaborative}, while the optimal value of $\alpha$ in (\ref{eq:our_reward}) is chosen in the range of: $\left\{ {0.3,0.5,0.7} \right\} $. We select the top-10 most similar users for each target user when computing the RACS reward function in (\ref{eq:our_reward}). All experiments are carried out with Intel (R) 12-Core (TM) E5-1650 v4 CPUs @ 3.60 GHz and GPU of NVIDIA GeForce RTX 3080. More implementation details are described in Appendix C.I.


\subsection{Results and Analyses}

In {\bf RQ1}, {\bf RQ3}, {\bf RQ5}, and {\bf RQ6}, we present experimental results on both pre-trained DiffRec \cite{wang2023diffusion} and CF-Diff \cite{hou2024collaborative}. In {\bf RQ2}, we show only the results of fine-tuning on the pre-trained CF-Diff due to space limitations, since those on the pre-trained DiffRec showed a similar tendency. We refer to Appendix C.III for more results of fine-tuning on the pre-trained DiffRec and CF-Diff.


\subsubsection{Comparison with competitors ({\bf RQ1})}

We validate the superiority of \textsf{ReFiT} over ten recommendation competitors for standard CF through extensive experiments on the three benchmark datasets. We evaluate the performance of fine-tuned models from DiffRec and CF-Diff, namely \textsf{ReFiT(DiffRec)} and \textsf{ReFiT(CF-Diff)}, respectively. Table \ref{tab:comparison} summarizes the results, and we make the following insightful observations.
\begin{enumerate}
    \item \textsf{ReFiT} {\it consistently} and {\it significantly} outperforms all recommendation competitors regardless of the datasets and the performance metrics, as confirmed by paired \textit{t}-tests showing statistically significant improvements ($p < 0.05$). The recommendation accuracy achieved by \textsf{ReFiT} exhibits standard deviations of 0.0014, 0.0006, and 0.0035 on average for the ML-1M, Yelp2018, and Anime datasets, respectively, demonstrating stable training performance. This stability is attributed to the fact that fine-tuning on a pre-trained model provides a well-initialized action space for RL. The standard deviations of the recommendation accuracy for all competing models are reported in Appendix C.III.1.
    \item \textsf{ReFiT(DiffRec)} and \textsf{ReFiT(CF-Diff)} consistently exhibit better performance than those of their counterparts, {\it i.e.}, DiffRec and CF-Diff, respectively. The gains can be attributed to the RL-aided fine-tuning strategy guided by our sophisticatedly designed reward function $r_\text{RACS}$, which enables \textsf{ReFiT} to fully exploit user--item interactions for personalized recommendations alongside collaborative signals.
    \item The performance gap between \textsf{ReFiT(DiffRec)} and DiffRec is the largest when the ML-1M dataset is used; the maximum improvement rate of 2.65\% is achieved in terms of N@20. 
    \item Diffusion-based recommender systems, DiffRec and CF-Diff, are superior to other generative model-based recommendation methods, including CFGAN, MultiDAE, and RecVAE. This is because diffusion-based recommender systems more intricately recover user--item interactions for recommendations due to their complex training nature. 
    \item Diffusion-based recommender systems are superior to graph convolution-based CF methods ({\it i.e.}, NGCF and LightGCN). This is attributed to better alignment with the generation process of real-world user--item interactions.
    \item RL-based CF methods ({\it i.e.}, FCPO and NICF) exhibit poor performance on larger datasets, Yelp2018 and Anime, primarily due to their large action spaces that complicate the search for higher-quality recommendations using RL strategies. Our \textsf{ReFiT} framework avoids this issue because it involves fine-tuning a pre-trained model, which operates within a well-initialized search space.
\end{enumerate}

\subsubsection{Impact of our reward function ({\bf RQ2})}
\label{sec:RQ2}

To discover whether our reward function $r_\text{RACS}({\bf u}_0)$ is indeed influential, we present its two variants:
\begin{itemize}
    \item Reward na\"ively using the top-$K$ recommendation accuracy ({\it i.e.}, $\alpha=1$): This reward represents the true positive recommended items within top-$K$ recommendations of only a given user and is expressed as $ r_\text{RA}\left( {{\bf u}_0 } \right) = N_K$.
    \item Reward using the cosine similarity: This reward is based on the cosine similarity between the historical user--item interactions ${\bf u}_0$ and the predicted interactions $\hat{\bf u}$ and is expressed as $r_\text{cos}\left( {{\bf u}_0 } \right) = sim\left( {{\bf u}_0 ,{\hat {\bf u}} } \right)$.
\end{itemize}

\begin{table}[t]\centering
\small
  \captionsetup{skip=0pt}
  \caption{Performance comparison among three reward functions. Here, the best and second-best performers are highlighted by bold and underline, respectively.}
  \label{tab:ablation}
  \begin{tabular}{cc|cccc}
    \toprule[1pt]
    \cmidrule{1-6}
           {\bf Dataset}& {\bf Reward} & {\bf R@10}& {\bf R@20} & {\bf N@10} & {\bf N@20} \\
    \midrule[1pt]
    \multirow{3}*{\textbf{ML-1M}}
    & $r_\text{RACS}({\bf u}_0)$& {\bf 0.1103}& {\bf 0.1866}& {\bf 0.0927}& {\bf 0.1185}\\
    & $r_\text{RA}({\bf u}_0)$ & 0.1089& \underline{0.1851}& \underline{0.0919}& \underline{0.1180}\\
    & $r_\text{cos}\left( {{\bf u}_0 } \right)$ & \underline{0.1095}& 0.1846& 0.0917& 0.1174\\
    \midrule[1pt]
    \multirow{3}*{\textbf{Yelp2018}}
    & $r_\text{RACS}({\bf u}_0)$& {\bf 0.0367}& {\bf 0.0618}& {\bf 0.0428}& {\bf 0.0516}\\
    & $r_\text{RA}({\bf u}_0)$ & 0.0364 & 0.0609 & \underline{0.0426} & \underline{0.0513} \\
    & $r_\text{cos}\left( {{\bf u}_0 } \right)$ & \underline{0.0365}& \underline{0.0614}& 0.0423& 0.0512\\
    \midrule[1pt]
    \multirow{3}*{\textbf{Anime}} 
    &$r_\text{RACS}({\bf u}_0)$ & {\bf 0.2283}& {\bf 0.3303}& {\bf 0.5319}& {\bf 0.4921}\\
    & $r_\text{RA}({\bf u}_0)$ & \underline{0.2279} & 0.3293 & \underline{0.5311} & \underline{0.4912} \\
    & $r_\text{cos}\left( {{\bf u}_0 } \right)$ & 0.2276& \underline{0.3299}& 0.5304& 0.4897\\
    \bottomrule[1pt]
  \end{tabular}
  \label{tab:diff_reward}
\end{table}


Table \ref{tab:diff_reward} summarizes the results on the three benchmark datasets with respect to all the metrics, and Fig. \ref{fig:reward_cf_diff} illustrates the behavior of three different reward functions during fine-tuning iterations. The reward values are derived from fine-tuning the pre-trained CF-Diff and have been normalized to a 0--1 range using min-max scaling to ensure comparability. Our observations are as follows:
\begin{enumerate}
    \item From Table \ref{tab:diff_reward}, the reward function $r_\text{RACS}({\bf u}_0)$ in (\ref{eq:our_reward}) always exhibits substantial gains over other variants, which demonstrates that incorporation of collaborative signals into the reward function is indeed beneficial in enhancing the recommendation accuracy.
    \item From Table \ref{tab:diff_reward}, the reward function $r_\text{RA}({\bf u}_0)$ is likely to outperform $r_\text{cos}({\bf u}_0)$ for most cases, which is attributed to the fact that $r_\text{RA}({\bf u}_0)$ is capable of inherently measuring the correctly recommended items while $r_\text{cos}({\bf u}_0)$ pays attention to the quality of reconstruction of user--item interactions using diffusion models.
    \item From Fig. \ref{fig:reward_cf_diff}, all rewards tend to increase consistently with the number of iterations. Notably, using $r_\text{RACS}({\bf u}_0)$ shows a tendency to converge faster than other reward functions, highlighting the importance of collaborative signals during fine-tuning for recommendations. 
\end{enumerate}

\begin{figure}[!t]
\centering
\begin{minipage}[b]{0.32\linewidth}
    \centering
    \includegraphics[width=\linewidth]{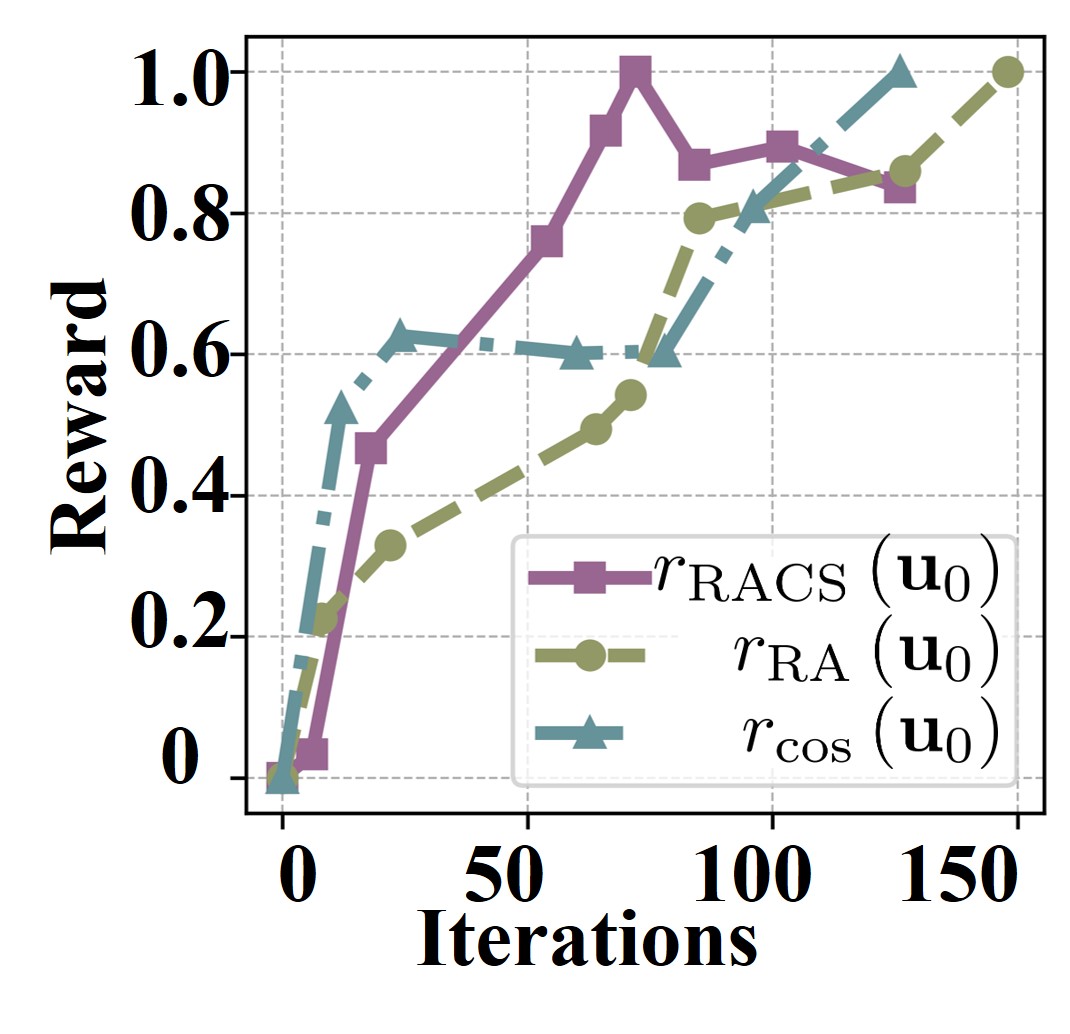}
    {\small\textbf{(a)} ML-1M}
\end{minipage}
\hfill
\begin{minipage}[b]{0.32\linewidth}
    \centering
    \includegraphics[width=\linewidth]{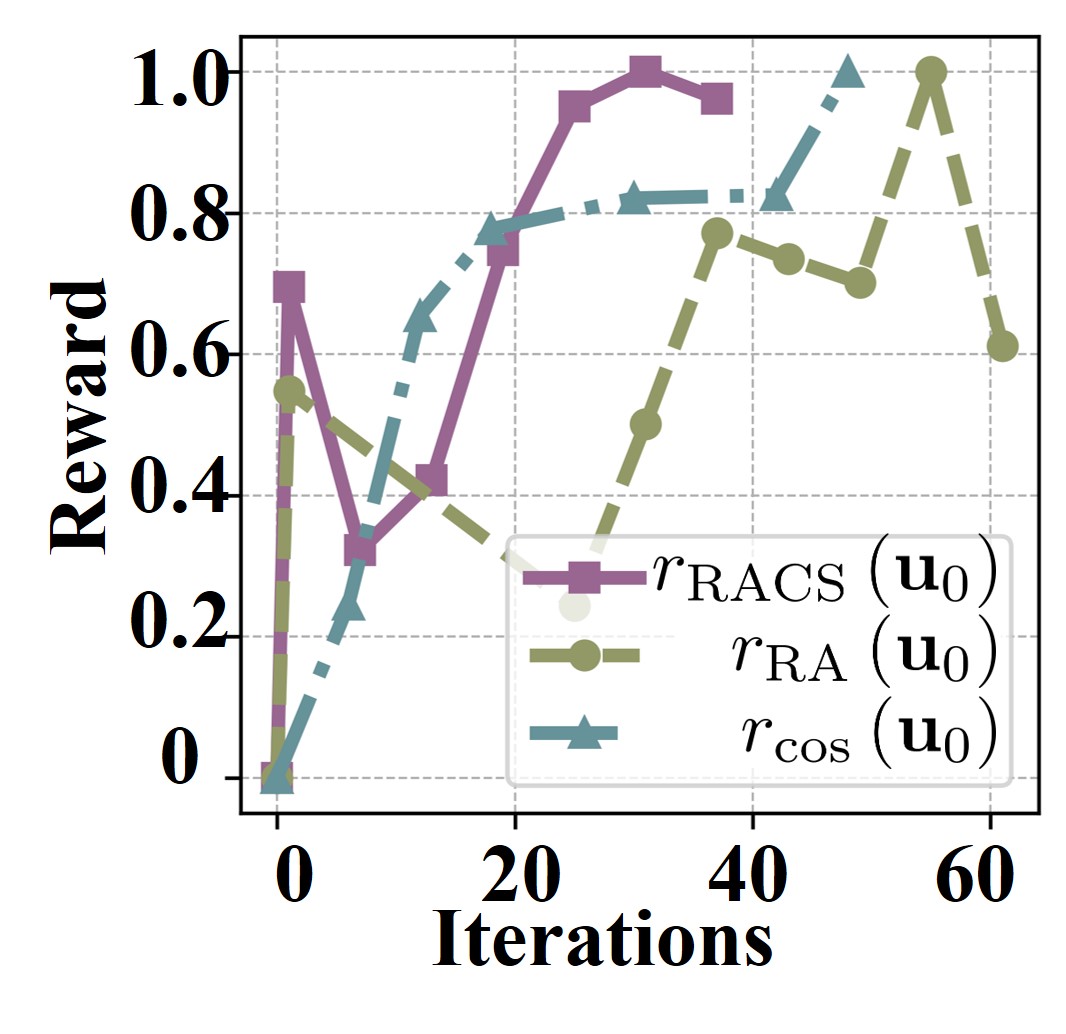}
    {\small\textbf{(b)} Yelp2018}
\end{minipage}
\hfill
\begin{minipage}[b]{0.32\linewidth}
    \centering
    \includegraphics[width=\linewidth]{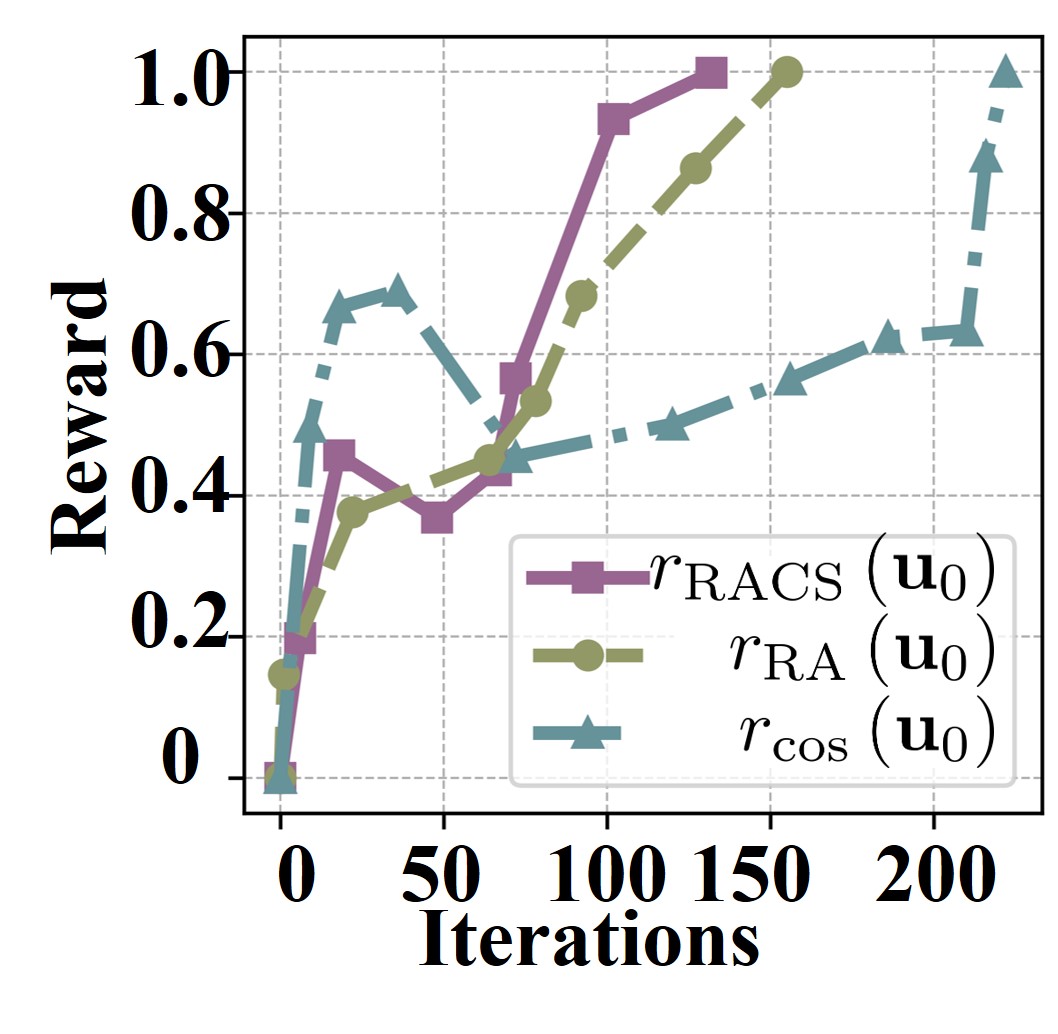}
    {\small\textbf{(c)} Anime}
\end{minipage}
\captionsetup{skip=0pt}
\caption{The behavior of different reward functions over iterations during fine-tuning given the pre-trained CF-Diff.}
\label{fig:reward_cf_diff}
\end{figure}


\subsubsection{Effectiveness of our RL-aided fine-tuning ({\bf RQ3})}
\label{sec:com_RL_fine_tuning}

 \begin{table}[t]\centering
\setlength\tabcolsep{0.9pt}
\small
  \captionsetup{skip=0pt}
  \caption{Performance comparison on the ML-1M, Yelp2018, and Anime among different fine-tuning strategies as well as the pre-trained model. Here, the best performer is highlighted by bold.}
  \label{tab:ablation}
  \begin{tabular}{c|c|cc|cc}
    \toprule[1pt]
    & \multicolumn{1}{c|}{}&\multicolumn{2}{|c|}{{\bf DiffRec}}&\multicolumn{2}{c}{{\bf CF-Diff}}\\
    \cmidrule{1-6}
           {\bf Dataset}&{\bf Method} & {\bf R@20}& {\bf N@20} & {\bf R@20} & {\bf N@20} \\
    \midrule[1pt]
    \multirow{5}*{\rotatebox{90}{{\bf ML-1M}}} &{\bf Pre-trained} & 0.1781 & 0.1131 & 0.1843 & 0.1176 \\ 
    \cmidrule{2-6}
    & {\bf ELBO-based fine-tuning} & 0.1776& 0.1127& 0.1821& 0.1169\\
    & {\bf RWR-based fine-tuning} & 0.1789& 0.1142 & 0.1851& 0.1179\\
    & {\bf PPO-based fine-tuning} & 0.1792& {\bf 0.1163}& 0.1861& 0.1184\\
    & {\bf REINFORCE-based (\textsf{ReFiT})} & {\bf 0.1799}& 0.1161& {\bf 0.1866}& {\bf 0.1185}\\
    \midrule[1pt]
    \multirow{5}*{\rotatebox{90}{{\bf Yelp2018}}} 
    & {\bf Pre-trained} & 0.0597 & 0.0499 & 0.0608 & 0.0509 \\ 
    \cmidrule{2-6}
    & {\bf ELBO-based fine-tuning} & 0.0563& 0.0452& 0.0579& 0.0494\\
    & {\bf RWR-based fine-tuning} & 0.0586& 0.0487 & 0.0598& 0.0501\\
    & {\bf PPO-based fine-tuning} & 0.0599& 0.0501& 0.0614& 0.0515\\
    & {\bf REINFORCE-based (\textsf{ReFiT})} & {\bf 0.0602}& {\bf 0.0504}& {\bf 0.0618}& {\bf 0.0516}\\
    \midrule[1pt]
    \multirow{5}*{\rotatebox{90}{{\bf Anime}}} 
    & {\bf Pre-trained} & 0.3249 & 0.4845 & 0.3265 & 0.4873 \\ 
    \cmidrule{2-6}
    & {\bf ELBO-based fine-tuning} & 0.3217& 0.4711& 0.3229& 0.4835\\
    & {\bf RWR-based fine-tuning} & 0.3254& 0.4853 & 0.3277& 0.4887\\
    & {\bf PPO-based fine-tuning} & {\bf 0.3268}& 0.4860& 0.3298& 0.4919\\
    & {\bf REINFORCE-based (\textsf{ReFiT})} & 0.3266& {\bf 0.4861}& {\bf 0.3303}& {\bf 0.4921}\\
    \bottomrule[1pt]
  \end{tabular}
  \label{tab:comp_FT}
  \vspace{-1em}
\end{table}

To investigate the effectiveness of our RL-aided fine-tuning (REINFORCE-based fine-tuning), we compare its performance against fine-tuning diffusion-based recommender systems using the ELBO-based loss (see Appendix A.I for the algorithmic details) as well as two variants of \textsf{ReFiT} that employ our RACS reward function, alongside reward-weighted regression (RWR)-based fine-tuning \cite{black2023training} and proximal policy optimization (PPO)-based fine-tuning \cite{black2023training}. As shown in Table \ref{tab:comp_FT}, our observations are as follows:

\begin{enumerate}
    \item \textsf{ReFiT} outperforms RWR-based fine-tuning and ELBO-based fine-tuning. Compared to RWR-based fine-tuning, \textsf{ReFiT}, as a policy gradient method using REINFORCE \cite{williams1992simple}, avoids inefficient updates caused by low-reward actions. Compared to ELBO-based fine-tuning, \textsf{ReFiT} mitigates the overfitting by leveraging RL to optimize the log-likelihood of user--item interactions directly through $\mathcal{J}_{RL} \left( {\theta  } \right)$ in (\ref{RL_objective}) instead, which supports our theoretical claim in Remark 2. 
    \item RWR-based fine-tuning reveals limited improvement as it prioritizes high-reward actions but still updates on low-reward ones, while wasting resources and potentially leading to suboptimal performance, as noted in \cite{black2023training}.
    \item ELBO-based fine-tuning tends to be even inferior to the pre-trained model, mainly due to the overfitting issue when the model keeps updating even after convergence, as illustrated in Fig.~\ref{fig:motivations}c.
    \item PPO-based fine-tuning performs comparably to \textsf{ReFiT} with REINFORCE, as the clipping mechanism in PPO has minimal impact when starting from a pre-trained model providing a well-initialized action space, as also shown in \cite{black2023training}.
\end{enumerate}


\begin{table}[t]\centering
\setlength\tabcolsep{3.0pt}
\small
\captionsetup{skip=0pt}
\caption{Performance comparison between a base model with no fine-tuning and \textsf{ReFiT} across sequential recommendation, social recommendation, and POI recommendation tasks. Here, the best performer is highlighted by bold.}
\label{tab:refit_all}
\begin{tabular}{ccccccc}
\toprule[1pt]
\textbf{Task} & \textbf{Dataset} & \textbf{Method} & \textbf{R@10} & \textbf{R@20} & \textbf{N@10} & \textbf{N@20} \\
\midrule[1pt]
\multirow{6}{*}{\rotatebox{90}{\textbf{Sequential Rec.}}}
    & \multirow{2}{*}{{\bf YooChoose}}  & {\bf DreamRec} & 0.0292 & 0.0493 & 0.0168 & 0.0223 \\
    &            & {\bf \textsf{ReFiT}} & \textbf{0.0397} & \textbf{0.0648} & \textbf{0.0225} & \textbf{0.0304} \\
    \cmidrule{2-7}
    & \multirow{2}{*}{{\bf KuaiRec}}  & {\bf DreamRec}       & 0.0196 & 0.0243 & {\bf 0.0167} & 0.0174 \\
    &            & {\bf \textsf{ReFiT}} & \textbf{0.0261} & \textbf{0.0311} & 0.0154 & \textbf{0.0189} \\
    \cmidrule{2-7}
    & \multirow{2}*{{\bf Zhihu}} 
    & {\bf DreamRec}& 0.0097& 0.0221& 0.0041& 0.0061\\
    & & {\bf \textsf{ReFiT}} & {\bf 0.0108} & {\bf 0.0237}& {\bf 0.0054}& {\bf 0.0067}\\
\midrule
\multirow{6}{*}{\rotatebox{90}{\textbf{Social Rec.}}}
    & \multirow{2}{*}{{\bf Ciao}}       & {\bf RecDiff}        & 0.0421 & 0.0700 & 0.0315 & 0.0410 \\
    &            & {\bf \textsf{ReFiT}} & \textbf{0.0439} & \textbf{0.0717} & \textbf{0.0322} & \textbf{0.0416} \\
    \cmidrule{2-7}
    & \multirow{2}*{{\bf Yelp-sub}} 
    & {\bf RecDiff}& 0.0373& 0.0605& 0.0244 & 0.0314\\
    & &{\bf \textsf{ReFiT}} & {\bf 0.0376}& {\bf 0.0607}& {\bf 0.0246}& {\bf 0.0316}\\
    \cmidrule{2-7}
    & \multirow{2}{*}{{\bf Epinions}}  & {\bf RecDiff}        & 0.0278& 0.0448& 0.0286& 0.0335 \\
    &            & {\bf \textsf{ReFiT}} & {\bf 0.0281} & {\bf 0.0451}& {\bf 0.0288}& {\bf 0.0337} \\
\midrule
\multirow{6}{*}{\rotatebox{90}{\textbf{POI Rec.}}}
    & \multirow{2}{*}{{\bf Foursquare}} & {\bf Diff-POI}       & 0.4317 & 0.4783 & 0.3645 & 0.3762 \\
    &            & {\bf \textsf{ReFiT}} & \textbf{0.4364} & \textbf{0.4801} & \textbf{0.3679} & \textbf{0.3791} \\
    \cmidrule{2-7}
    & \multirow{2}{*}{{\bf TKY}}  & {\bf Diff-POI}       & 0.6668& 0.6958& 0.6119 & 0.6192 \\
    &            & {\bf \textsf{ReFiT}} & {\bf 0.6691}& {\bf 0.6972}& {\bf 0.6148}& {\bf 0.6225} \\
    \cmidrule{2-7}
    & \multirow{2}*{{\bf NYC}} 
    & {\bf Diff-POI}& 0.6861& 0.7014& 0.6606& 0.6645\\
    & &{\bf \textsf{ReFiT}} & {\bf 0.6894} & {\bf 0.7043}& {\bf 0.6637}& {\bf 0.6671}\\
\bottomrule[1pt]
\end{tabular}
\vspace{-1em}
\end{table}

\subsubsection{Other downstream recommendation tasks ({\bf RQ4})}
\label{exp:rq5}

Our \textsf{ReFiT} framework is broadly applicable to diverse datasets and recommendation tasks, as long as a diffusion model can be employed to generate user--item interactions. To demonstrate such generalization, we conduct extensive experiments on sequential recommendations (DreamRec \cite{yang2023generate}), social recommendations (RecDiff \cite{li2024recdiff}), and POI recommendations (Diff-POI \cite{qin2023diffusion}) in which diffusion models were developed in the literature. For each task, we evaluate performance on three real-world datasets, which are widely used for those tasks. Table \ref{tab:refit_all} summarizes the performance comparison between the base model (no fine-tuning) and \textsf{ReFiT} on three recommendation tasks. Our observations are as follows:


\begin{enumerate}
    \item \textsf{ReFiT} {\it consistently} and {\it significantly} outperforms the corresponding benchmark method regardless of the datasets and the performance metrics (except for the case of KuaiRec in N@10). This showcases the general superiority and robustness of \textsf{ReFiT} in various recommendation tasks.
    \item \textsf{ReFiT} achieves the {\bf maximum improvement rate of 36.3\%} in terms of N@20 on the YooChoose dataset for sequential recommendations. Compared to the case of recommendations for standard CF, this significant gain comes from the fact that the datasets are relatively sparse (see Seq. in Table~\ref{table:datasets}). In such sparse datasets, correctly recommending even one additional item can lead to substantial increases in Recall and NDCG, as both metrics are sensitive when the number of relevant items (or ground truth) is limited. 
\end{enumerate}

Additionally, Figs. 6--8 illustrate the Recall@20 and NDCG@20 performance across fine-tuning iterations using \textsf{ReFiT} for sequential, social, and POI recommendation tasks, respectively. The results demonstrate that the recommendation accuracy achieved by \textsf{ReFiT} consistently increases with the number of iterations.

\subsubsection{Sensitivity analysis ({\bf RQ5})}

\label{app:sensitivity}
Our method involves a single tunable parameter, $\alpha$ in (\ref{eq:our_reward}), which balances the collaborative signal-aware reward. We analyze its impact on the recommendation accuracy for all the datasets. From Fig. \ref{fig:alpha}, the maximum NDCG@20 is achieved at $\alpha = 0.5$ on ML-1M and Anime, and at $\alpha = 0.7$ on Yelp2018. It reveals that high values of $\alpha$ degrades the performance since collaborative signals along with multi-hop neighbors are passively utilized and low values of $\alpha$ overuses the effect of collaborative signals. Note that $\alpha = 1.0$ corresponds to the case of using the reward function $r_{RA}({\bf u}_0)$ in Section \ref{sec:RQ2}. Hence, it is crucial to suitably determine the value of $\alpha$ in guaranteeing the optimal performance. 

\begin{figure}[!t]
\centering
\begin{minipage}[b]{0.32\linewidth}
    \centering
    \includegraphics[width=\linewidth]{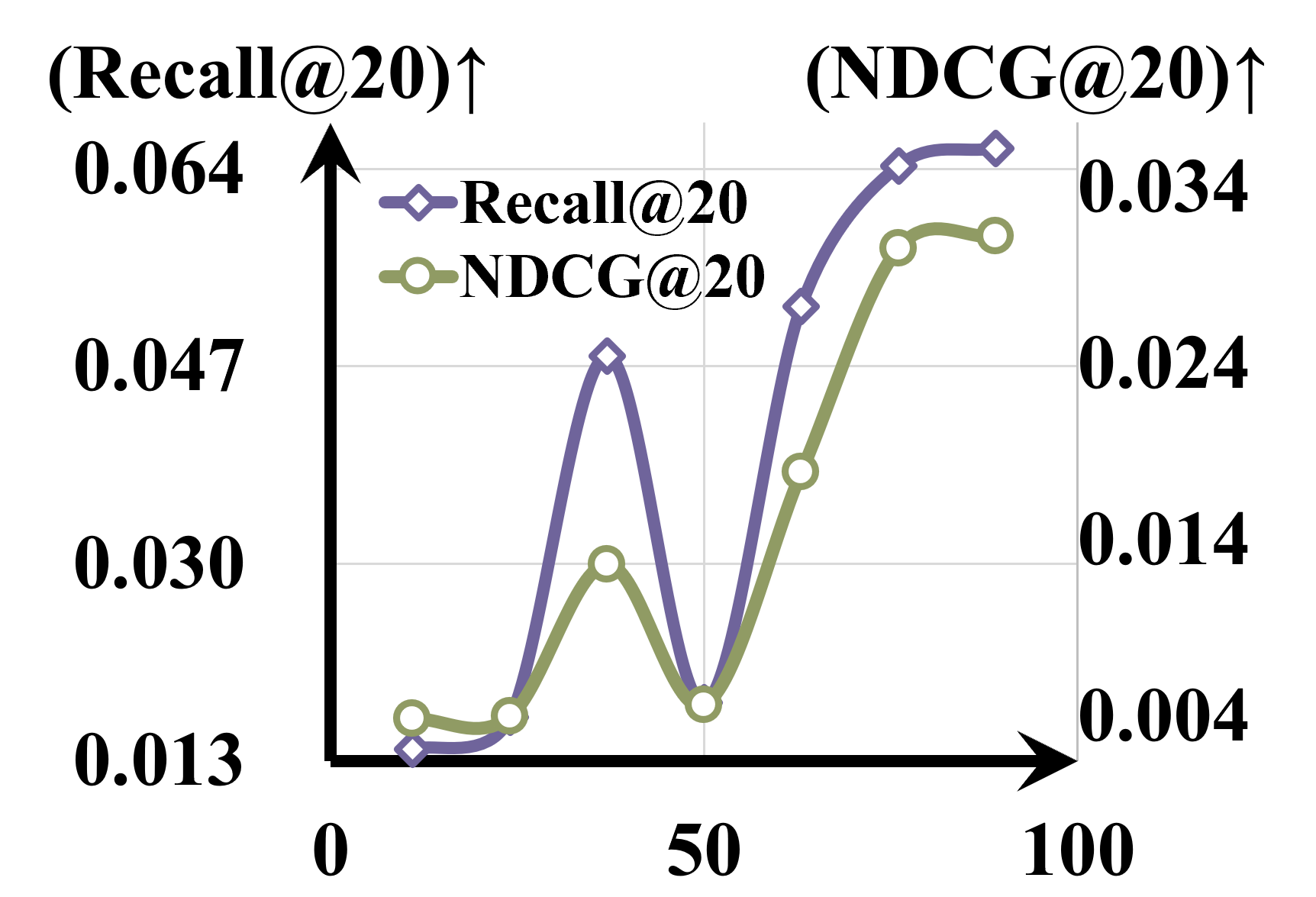}
    {\small\textbf{(a)} YooChoose}
\end{minipage}
\hfill
\begin{minipage}[b]{0.32\linewidth}
    \centering
    \includegraphics[width=\linewidth]{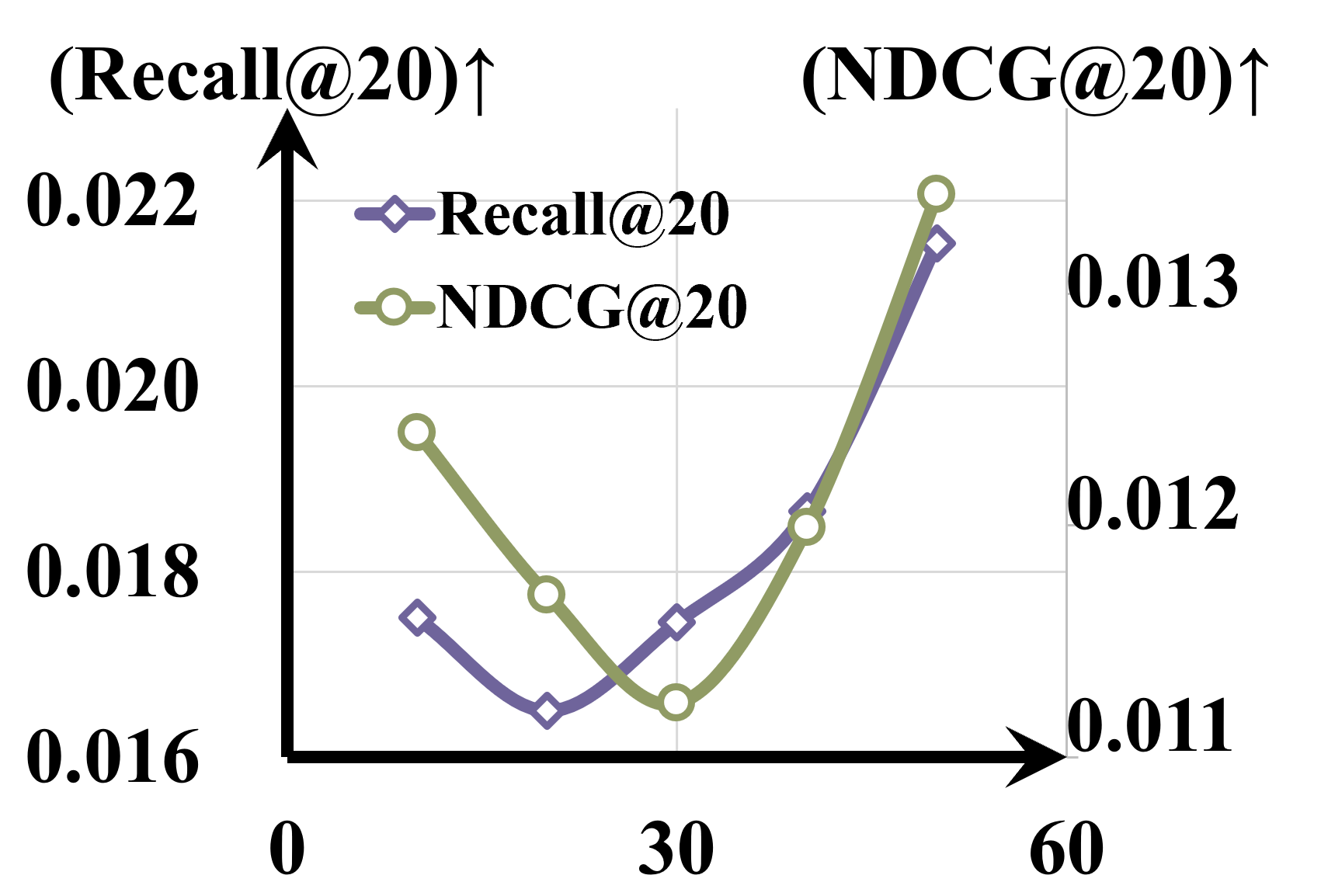}
    {\small\textbf{(b)} KuaiRec}
\end{minipage}
\hfill
\begin{minipage}[b]{0.32\linewidth}
    \centering
    \includegraphics[width=\linewidth]{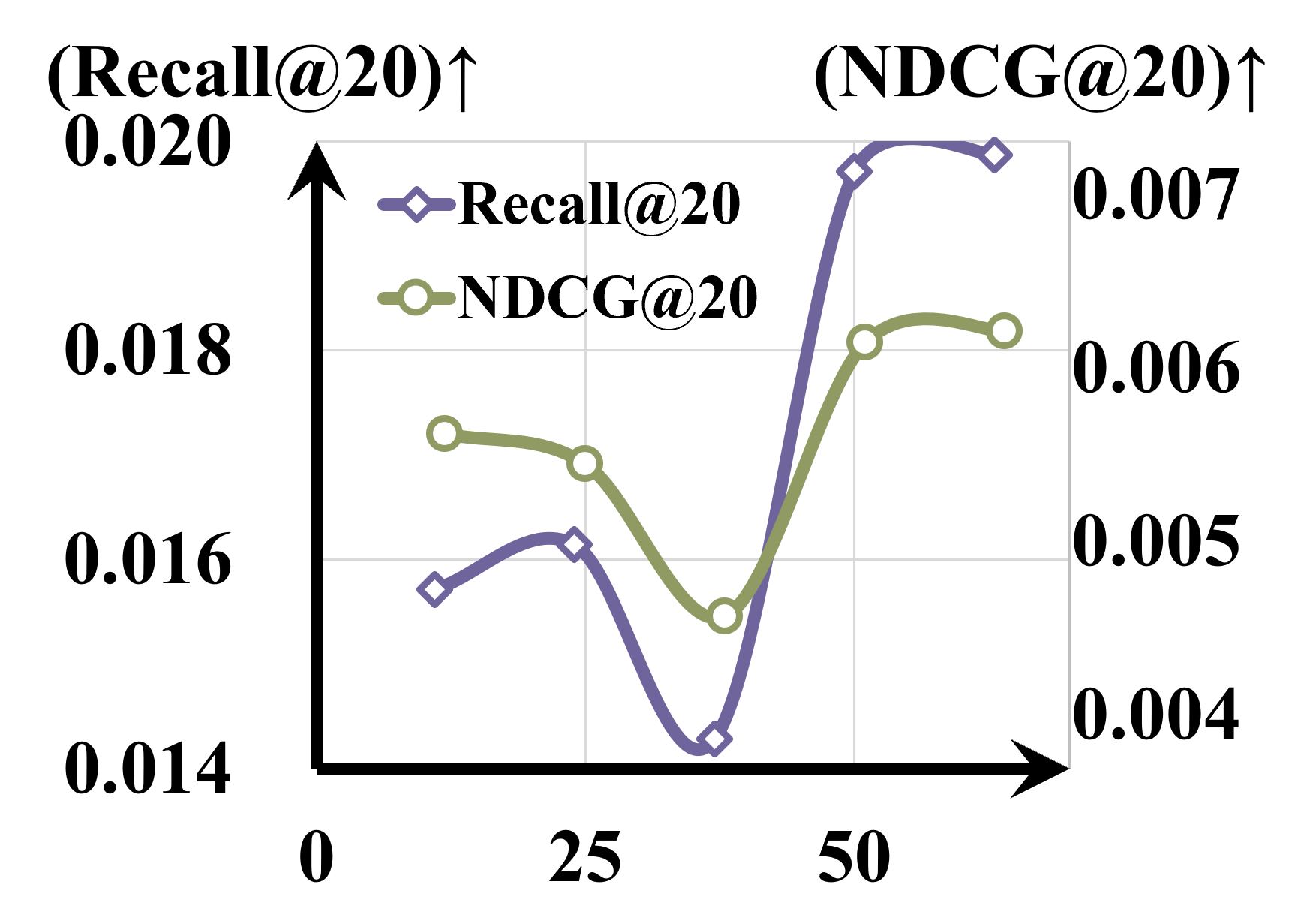}
    {\small\textbf{(c)} Zhihu}
\end{minipage}
\captionsetup{skip=0pt}
\caption{Recommendation accuracy in terms of Recall@20 and NDCG@20 across fine-tuning iterations for \textsf{ReFiT(DreamRec)}.}
\label{fig:ndcg_vs_iters_sq}
\end{figure}

\begin{figure}[!t]
\centering
\begin{minipage}[b]{0.32\linewidth}
    \centering
    \includegraphics[width=\linewidth]{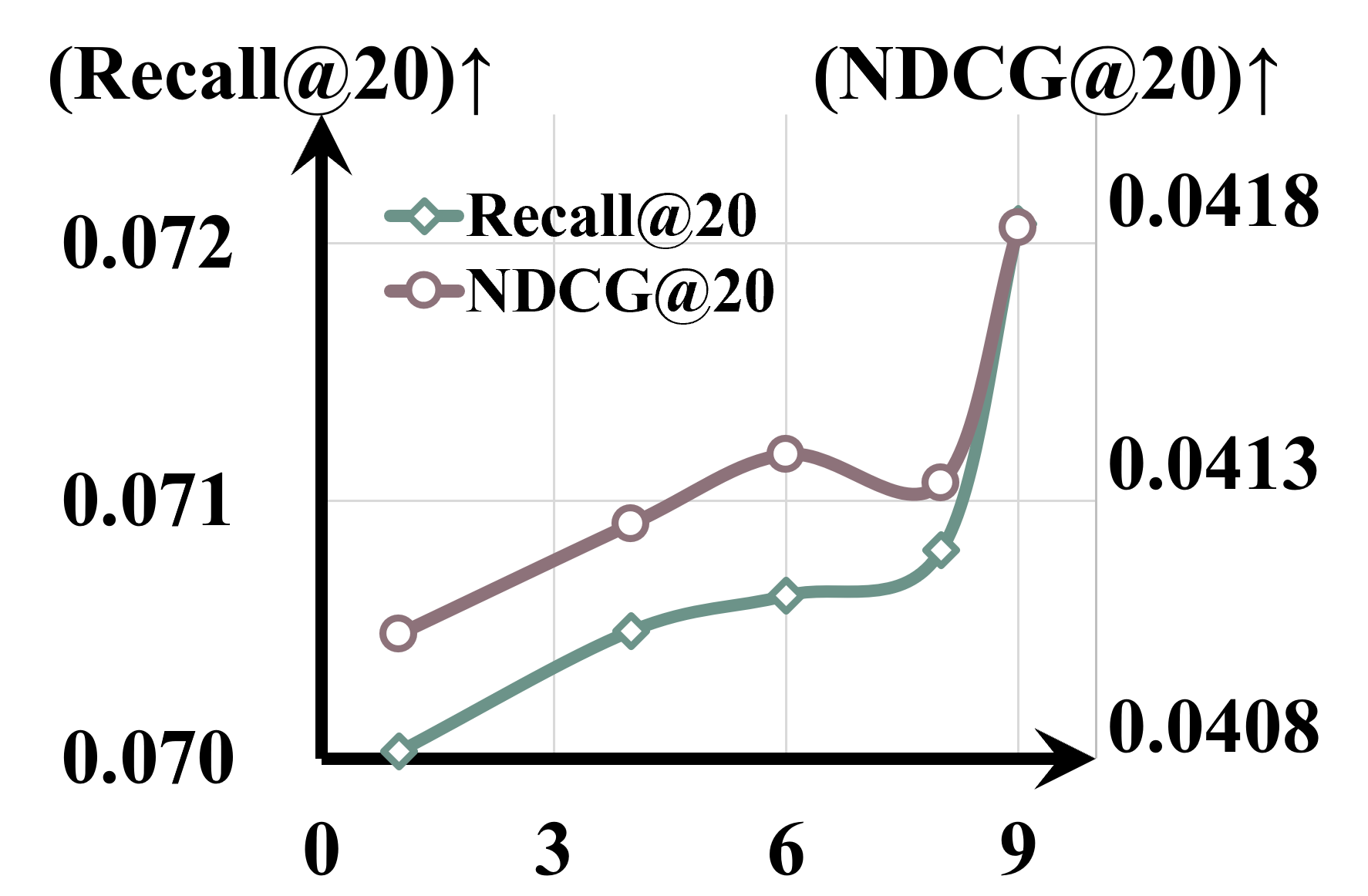}
    {\small\textbf{(a)} Ciao}
\end{minipage}
\hfill
\begin{minipage}[b]{0.32\linewidth}
    \centering
    \includegraphics[width=\linewidth]{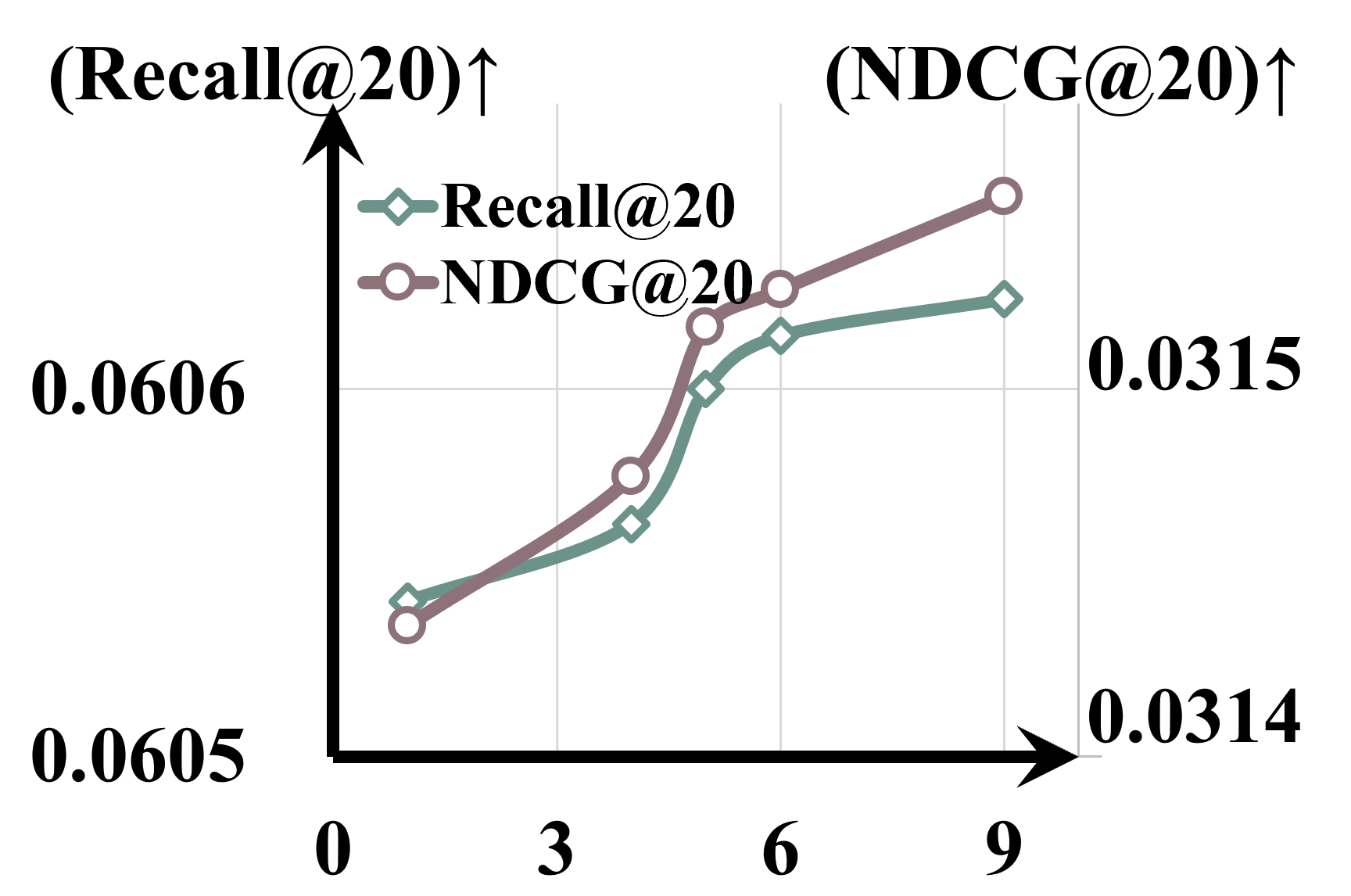}
    {\small\textbf{(b)} Yelp}
\end{minipage}
\hfill
\begin{minipage}[b]{0.32\linewidth}
    \centering
    \includegraphics[width=\linewidth]{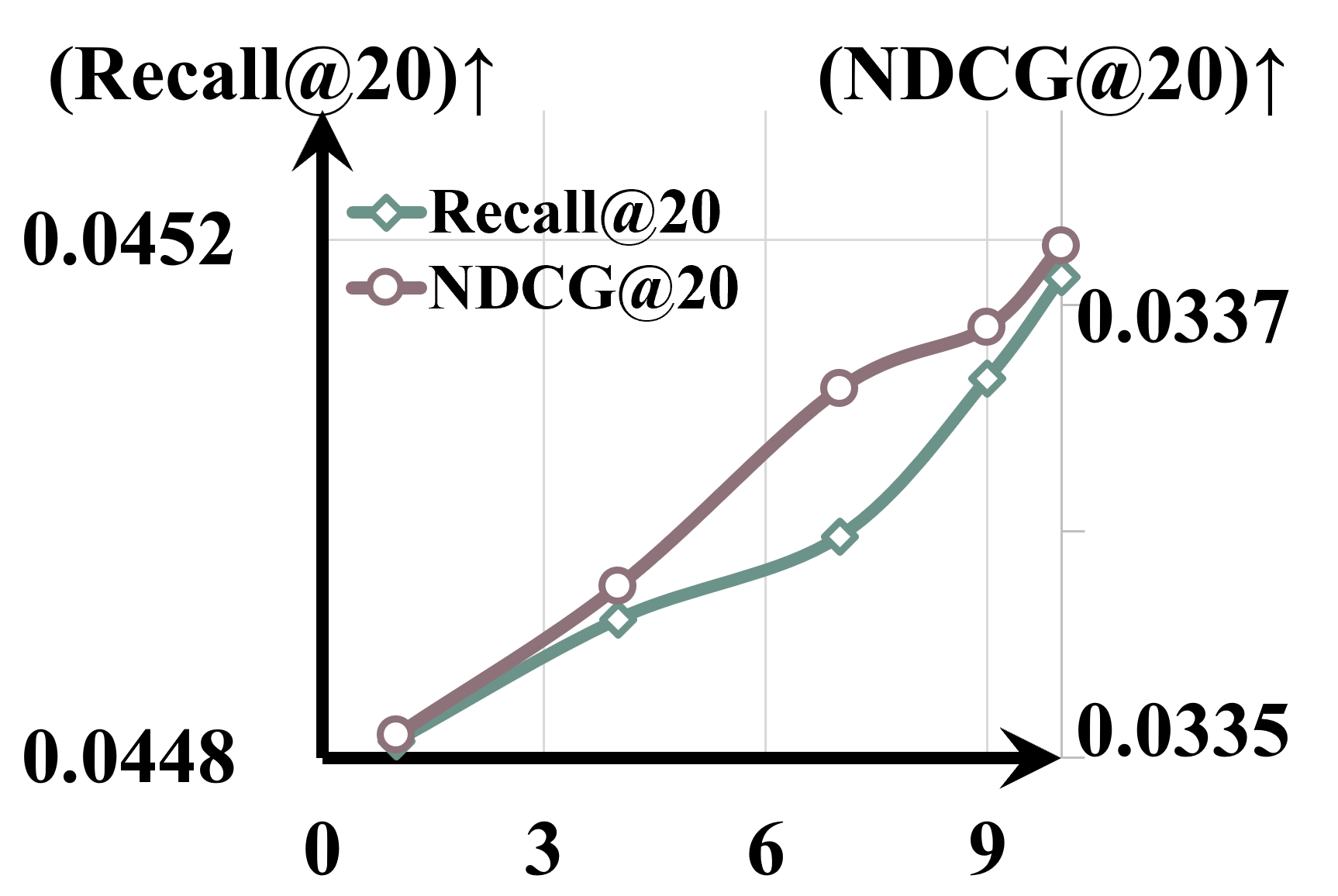}
    {\small\textbf{(c)} Epinions}
\end{minipage}
\captionsetup{skip=0pt}
\caption{Recommendation accuracy in terms of Recall@20 and NDCG@20 across fine-tuning iterations for \textsf{ReFiT(RecDiff)}.}
\label{fig:ndcg_vs_iters_social}
\end{figure}

\begin{figure}[!t]
\centering
\begin{minipage}[b]{0.32\linewidth}
    \centering
    \includegraphics[width=\linewidth]{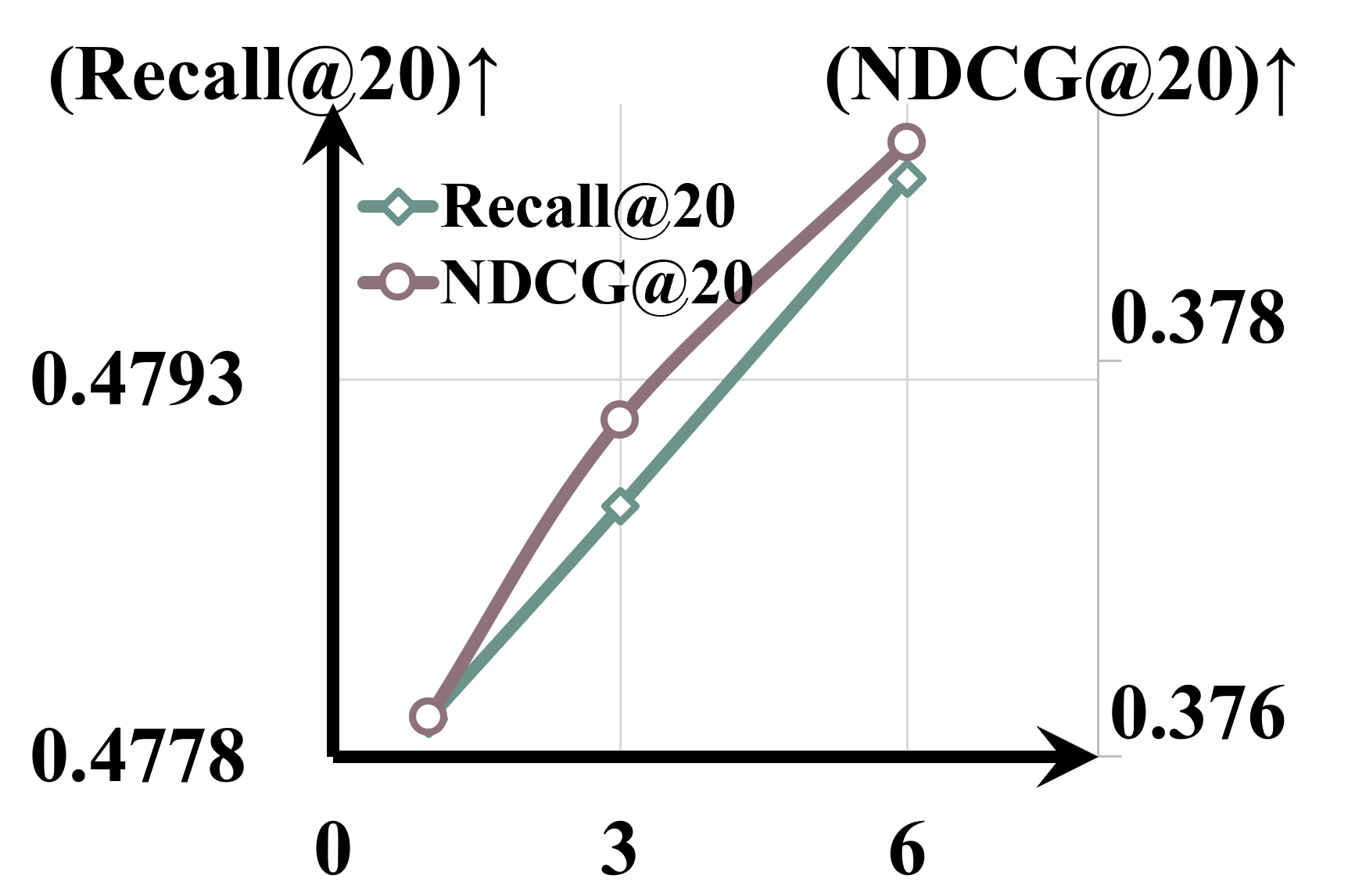}
    \vspace{-2pt}
    {\small\textbf{(a)} Foursquare}
\end{minipage}
\hfill
\begin{minipage}[b]{0.32\linewidth}
    \centering
    \includegraphics[width=\linewidth]{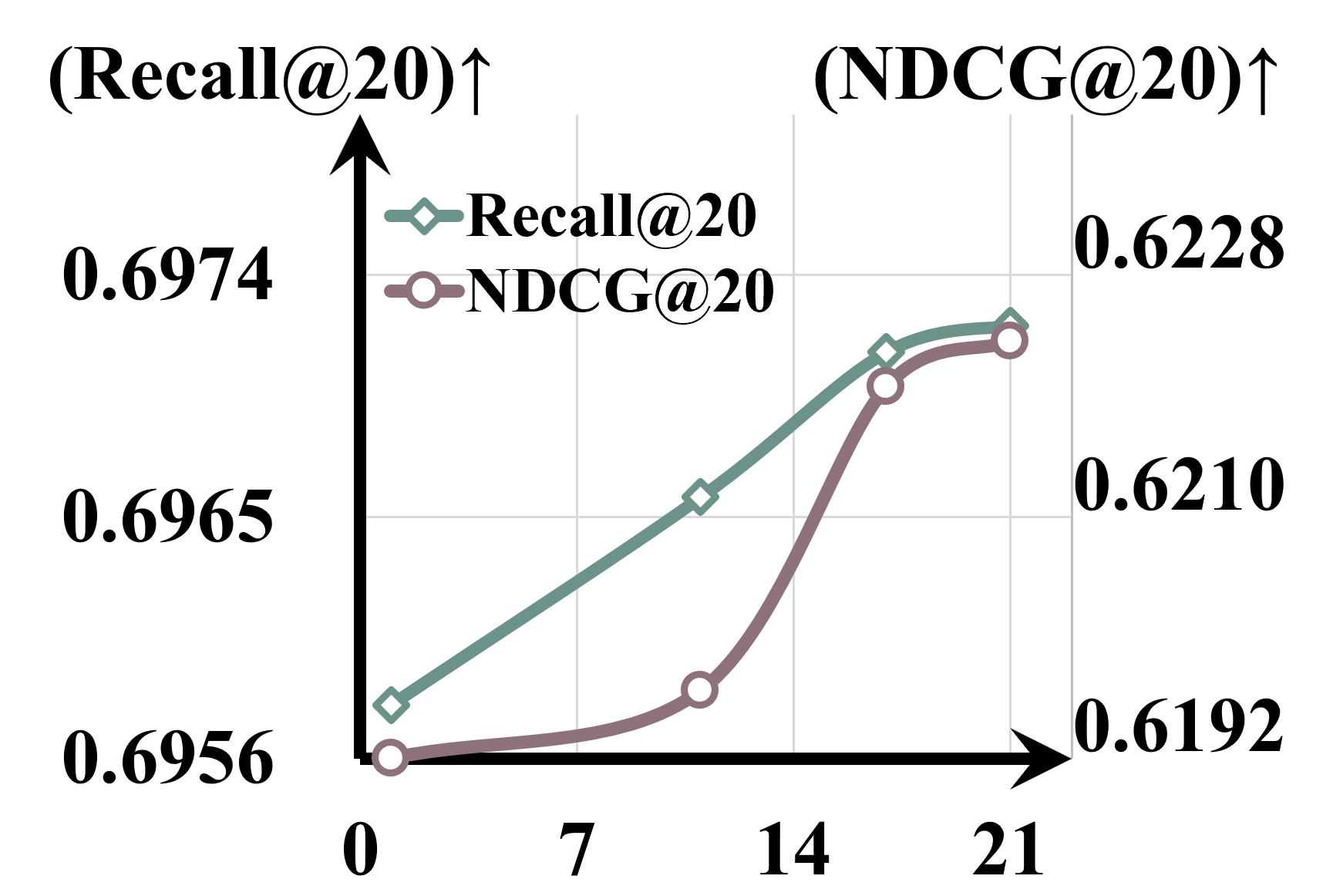}
    \vspace{-2pt}
    {\small\textbf{(b)} TKY}
\end{minipage}
\hfill
\begin{minipage}[b]{0.32\linewidth}
    \centering
    \includegraphics[width=\linewidth]{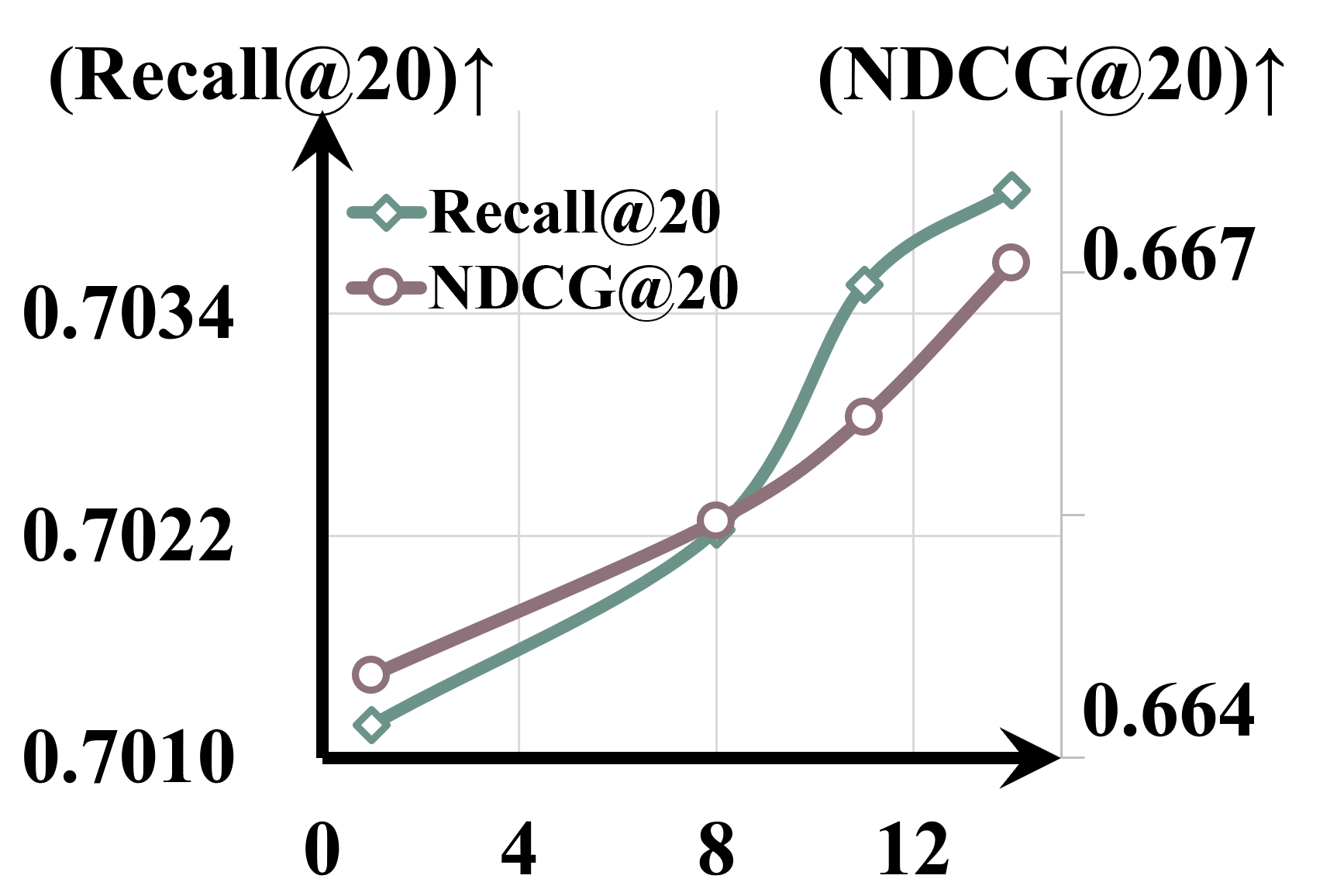}
    \vspace{-2pt}
    {\small\textbf{(c)} NYC}
\end{minipage}
\captionsetup{skip=0pt}
\caption{Recommendation accuracy in terms of Recall@20 and NDCG@20 across fine-tuning iterations for \textsf{ReFiT(Diff-POI)}.}
\label{fig:ndcg_vs_iters_poi}
\end{figure}

\begin{figure}[!t]
\centering
\begin{minipage}[b]{0.32\linewidth}
    \centering
    \includegraphics[width=\linewidth]{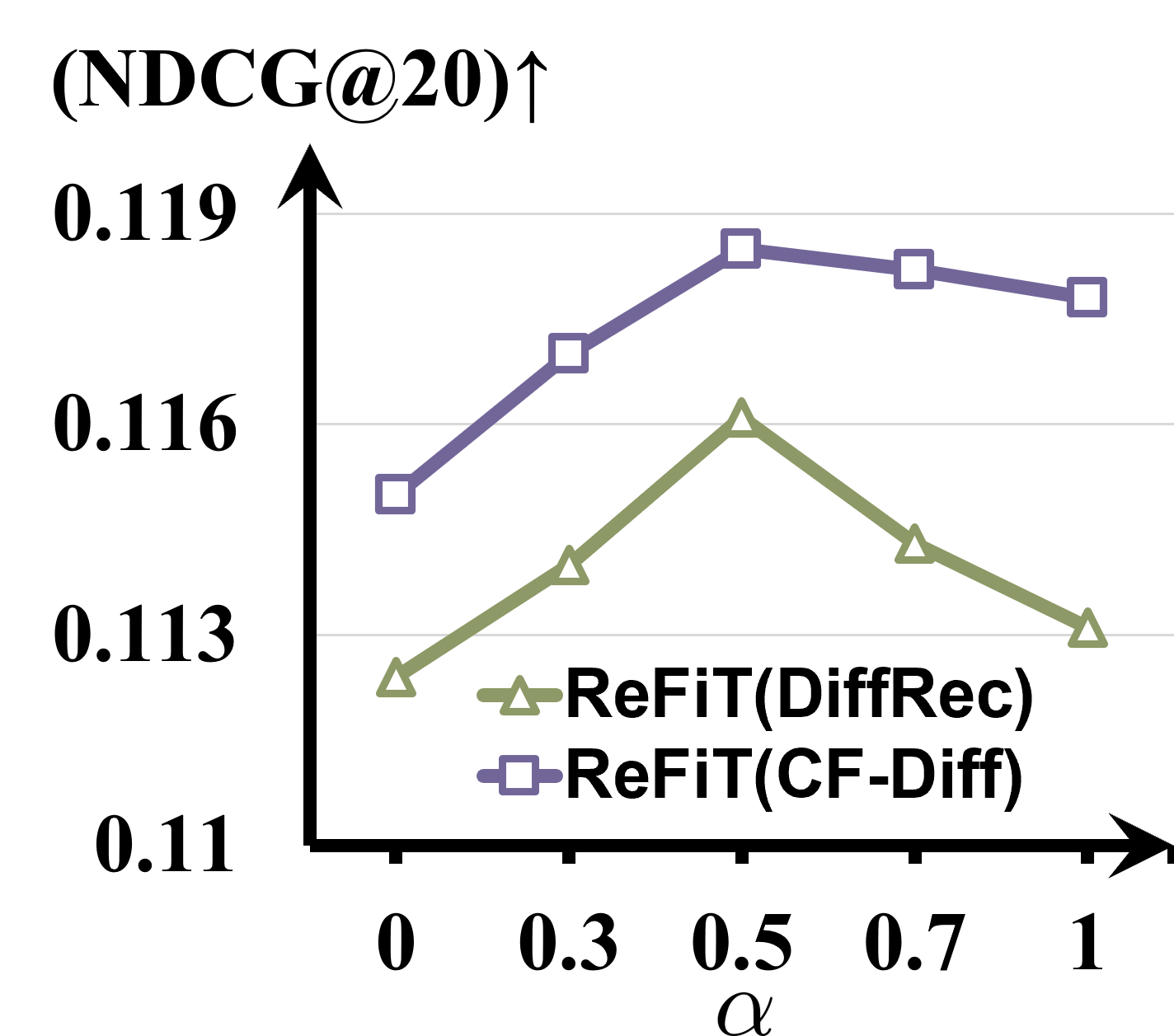}
    \vspace{-2pt}
    {\small\textbf{(a)} ML-1M}
\end{minipage}
\hfill
\begin{minipage}[b]{0.32\linewidth}
    \centering
    \includegraphics[width=\linewidth]{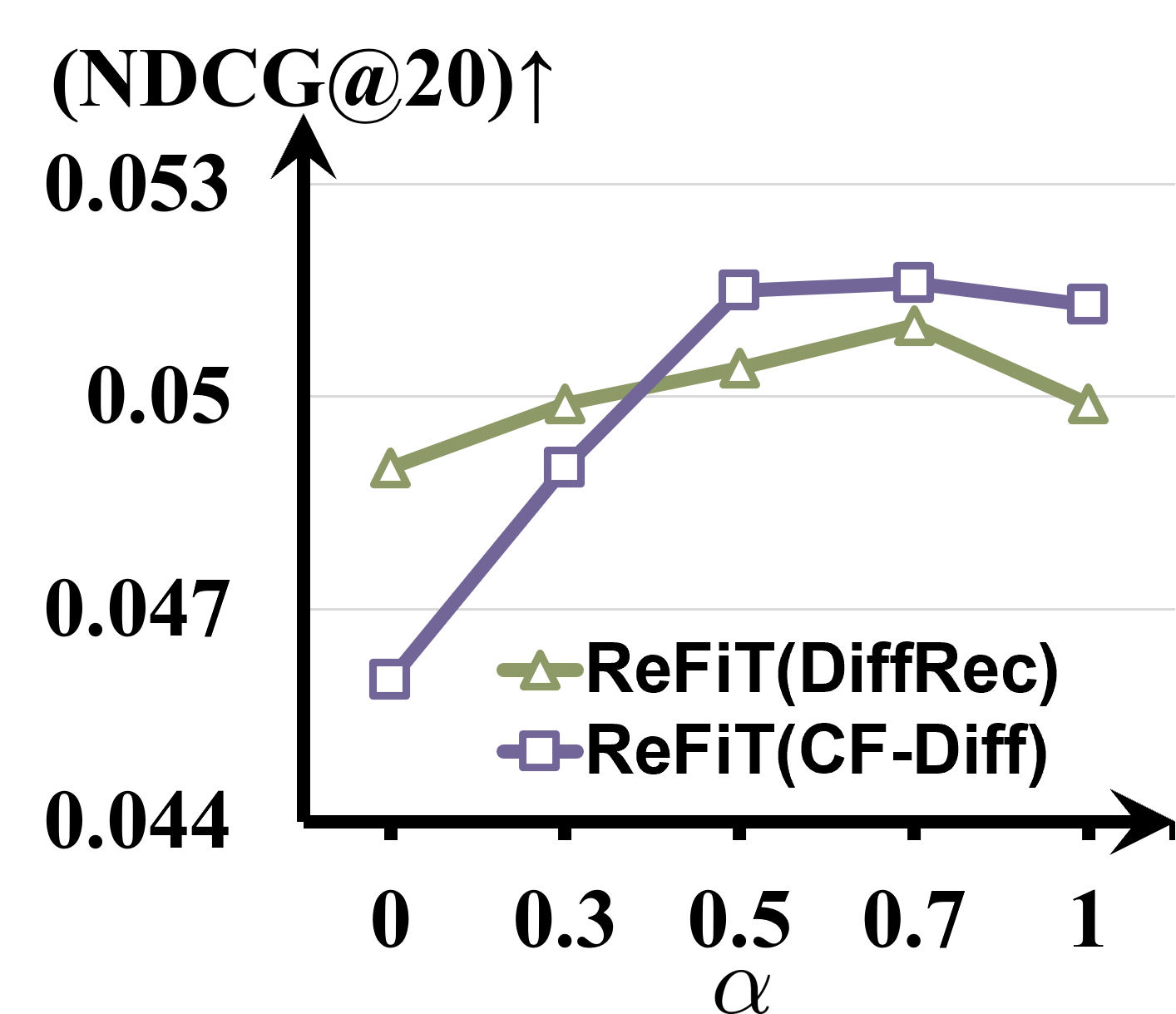}
    \vspace{-2pt}
    {\small\textbf{(b)} Yelp2018}
\end{minipage}
\hfill
\begin{minipage}[b]{0.32\linewidth}
    \centering
    \includegraphics[width=\linewidth]{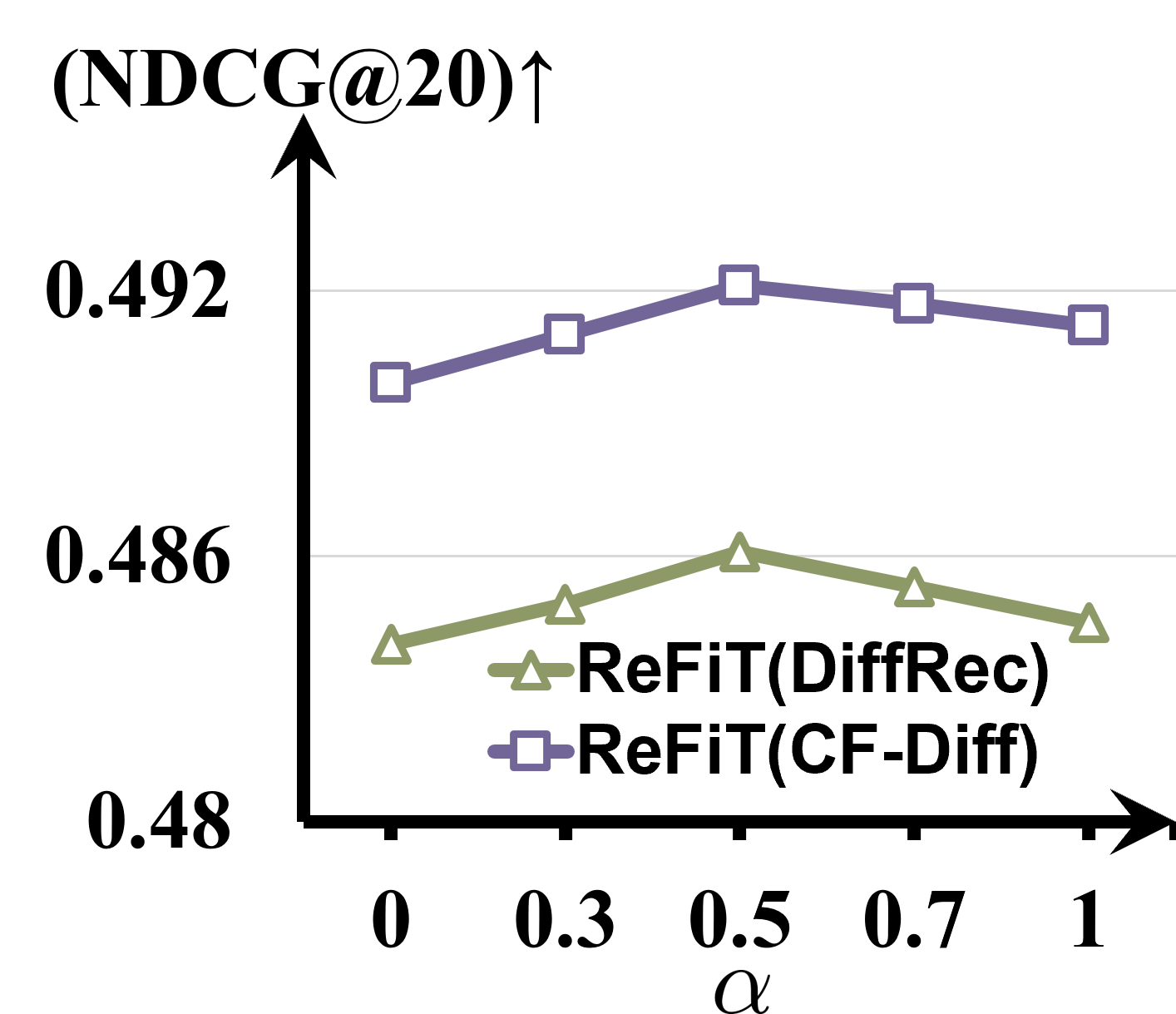}
    \vspace{-2pt}
    {\small\textbf{(c)} Anime}
\end{minipage}
\captionsetup{skip=0pt}
\caption{Effect of hyperparameter $\alpha$ on NDCG@20 for \textsf{ReFiT(DiffRec)} and \textsf{ReFiT(CF-Diff)}.}
\label{fig:alpha}
\end{figure}

    

\subsubsection{Space and time efficiency ({\bf RQ6})}
First, we compare the memory consumption of \textsf{ReFiT} with its corresponding diffusion-based recommender systems for pre-training. As shown in Table~\ref{tab:mem_com}, the VRAM usage primarily reflects the size of the model parameters, while RAM usage accounts for additional components such as data loading, optimizer states, and other implementation-specific buffers. Notably, \textsf{ReFiT} incurs only minimal overhead, with RAM and VRAM usage increasing by merely a few percent. From the fact that \textsf{ReFiT} fine-tunes a pre-trained model ({\it e.g.,} DiffRec or CF-Diff), we empirically confirm that its memory footprint remains almost identical to the pre-trained model. The only minor overhead arises from storing reward signals and sampled diffusion trajectories, which is negligible in practice.

\begin{table}[t]\centering
\setlength\tabcolsep{8pt}
\small
  \captionsetup{skip=0pt}
  \caption{Comparison of memory usage and relative increase ($\Delta$\%) between \textsf{ReFiT} and its corresponding pre-trained diffusion-based recommenders on the ML-1M dataset.}
  \resizebox{0.48\textwidth}{!}{
  \begin{tabular}{c|cc|c|cc|c}
    \toprule[1pt]
    \cmidrule{1-7}
    {\bf Resource} & {\bf DiffRec} & {\bf \textsf{ReFiT}} & {\bf $\Delta$\%} & {\bf CF-Diff} & {\bf \textsf{ReFiT}} & {\bf $\Delta$\%} \\
    \midrule[1pt]
    {\bf RAM} & 1489.39 MB & 1571.62 MB & +5.52\% & 1631.10 MB & 1754.07 MB & +7.54\% \\
    {\bf VRAM} & 29.13 MB & 30.84 MB & +5.87\% & 94.49 MB & 95.64 MB & +1.22\% \\
    \bottomrule[1pt]
  \end{tabular}
  }
  \label{tab:mem_com}
\end{table}

Second, we evaluate computational efficiency by comparing execution time (in seconds per iteration) and recommendation performance (R@10) across DiffRec, CF-Diff, and \textsf{ReFiT} on three datasets. As shown in Table~\ref{tab:times_app}, \textsf{ReFiT} improves the performance of diffusion-based recommenders with only a negligible time overhead---achieving up to 40× faster execution. For instance, on the ML-1M dataset, \textsf{ReFiT(CF-Diff)} improves R@10 by 2.41\% while reducing execution time by 95.18\%. This efficiency comes from our RL-aided fine-tuning strategy, which focuses on higher-reward samples in each iteration rather than training on the entire user set as in DiffRec and CF-Diff.


\begin{table}[t]\centering
\setlength\tabcolsep{1.1pt}
\small
  \captionsetup{skip=0pt}
  \caption{Comparison of performance (R@10) and per-iteration execution time (in seconds per iteration) between pre-trained diffusion-based recommenders (DiffRec and CF-Diff) and \textsf{ReFiT} on the ML-1M, Yelp2018, and Anime datasets.}
  \label{tab:ablation}
  \begin{tabular}{c|cc|cc|cc}
    \toprule[1pt]
    \multicolumn{1}{c|}{}&\multicolumn{2}{|c|}{{\bf ML-1M}}&\multicolumn{2}{c|}{{\bf Yelp2018}}&\multicolumn{2}{c}{{\bf Anime}}\\
    \cmidrule{1-7}
            {\bf Method} & {\bf R@10} & {\bf Time$\downarrow$} & {\bf R@10} & {\bf Time$\downarrow$}& {\bf R@10} & {\bf Time$\downarrow$}\\
    \midrule[1pt]
    {\bf DiffRec}& 0.1058& 3.44 & 0.0351& 18.02 & 0.2193& 30.94\\
    {\bf \textsf{ReFiT(DiffRec)}}& {\bf 0.1083}& {\bf 0.28} & {\bf 0.0355}&{\bf 0.43} & {\bf 0.2231}& {\bf 0.81}\\
    {\bf Gain} & 2.36\% & \cellcolor{lightgray} -91.86\% & 1.14\% & \cellcolor{lightgray} -97.61\%& 1.73\%& \cellcolor{lightgray} -97.37\%\\
    \midrule[1pt]
    {\bf CF-Diff}& 0.1077& 7.67 & 0.0363& 142.74 & 0.2263& 81.78\\ 
    {\bf \textsf{ReFiT(CF-Diff)}}& {\bf 0.1103}& {\bf 0.37} &{\bf 0.0367} & {\bf 6.14} & {\bf 0.2283}& {\bf 2.33}\\
    {\bf Gain} & 2.41\% & \cellcolor{lightgray} -95.18\% & 1.10\% & \cellcolor{lightgray} -95.70\%& 0.88\%& \cellcolor{lightgray} -97.15\%\\
    \bottomrule[1pt]
  \end{tabular}
  \label{tab:times_app}
\end{table}

Additionally, to empirically validate the scalability of \textsf{ReFiT}, we measure execution time on synthetic user--item interaction datasets generated at random with a sparsity of 0.99, analogous to that observed on Yelp2018 and Anime. We vary dataset sizes in two settings: (1) fixing $\left| \mathcal{I} \right| = 1e^4$ and increasing $\left| \mathcal{U} \right|\! \in\! \left\{ {1e^4\! ,3e^4 \!,6e^4\! ,7e^4 \!,8e^4\! ,9e^4 } \right\} $; and (2) fixing $\left| \mathcal{U} \right| = 1e^4$ and increasing $\left| \mathcal{I} \right|$ over the same range. Fig. \ref{fig:time}a ({\it resp.} Fig. \ref{fig:time}b) show the per-iteration execution time (in seconds) of \textsf{ReFiT} (fine-tuned on DiffRec), as the number of users ({\it resp.} the number of items) increases. The dashed line indicates a linear scaling in $\left| \mathcal{U} \right|$ and $\left| \mathcal{I} \right|$. It can be seen that our empirical evaluation concurs with the theoretical analysis in {\bf Theorem~\ref{sec:theorem2}}.

\begin{figure}[t!]
\pgfplotsset{footnotesize,samples=10}
\centering
\begin{tikzpicture}[scale=1]
    \begin{axis}[        
    xlabel= (a) $\left| \mathcal{U} \right|$, 
    ylabel= Execution time (s), 
    xlabel style={yshift=0.5em}, 
    ylabel style={yshift=-2em}, 
    grid=major, grid style={dashed}, width = 4.9cm, height = 4.0cm , legend style={at={(0.41,0.2)},anchor=west},
    xtick={10000,35000,65000,90000},
    scaled x ticks=false,
    xticklabel style={font=\scriptsize},
    xticklabels={$1 \cdot 10^4$,$3 \cdot 10^4$,$6 \cdot 10^4$,$9 \cdot 10^4$},
    xticklabel style={yshift=3pt},
    yticklabel style={font=\scriptsize},
    yticklabel style={xshift=3pt},
    ],
        \addplot [color=fig_8_color_1, mark=square, legend = \textsf{ReFiT}]
            coordinates {
                (10000, 0.3)(30000, 0.49)(60000, 0.68)(70000, 0.81)(80000, 0.90)(90000, 1.10)
            };
        \addplot [ dashed, color=red,        mark=none, legend = $$\mathcal{O}\left( {\left| \mathcal{U} \right|} \right)$$ ]
            coordinates {
                (5000,0.21)(90000,1.06)
            };
    \addlegendentry{\textsf{ReFiT}}
    \addlegendentry{$\mathcal{O}\left( {\left| \mathcal{U} \right|} \right)$}
    \end{axis}
\end{tikzpicture}
\begin{tikzpicture}[scale=1]
    \begin{axis}[        
    xlabel= (b) $\left| \mathcal{I} \right|$, 
    ylabel= Execution time (s), 
    xlabel style={yshift=0.5em}, 
    ylabel style={yshift=-2em}, 
    grid=major, grid style={dashed}, width = 4.9cm, height = 4.0cm, legend style={at={(0.43,0.2)},anchor=west},
    scaled x ticks=false,
    xtick={10000,35000,65000,90000},
    xticklabel style={font=\scriptsize},
    xticklabels={$1 \cdot 10^4$,$3 \cdot 10^4$,$6 \cdot 10^4$,$9 \cdot 10^4$},
    xticklabel style={yshift=3pt},
    yticklabel style={font=\scriptsize},
    yticklabel style={xshift=3pt},
    ]
        \addplot [color=fig_8_color_1, mark=square, legend = \textsf{CF-Diff}]
            coordinates {
                (10000, 0.3)(30000, 0.39)(60000,0.5)(70000, 0.60)(80000, 0.69)(90000, 0.76)
            };
        \addplot [ dashed, color=red,        mark=none, legend = $$\mathcal{O}\left( {\left| \mathcal{I} \right|} \right)$$ ]
            coordinates {
                (5000,0.23)(90000,0.74)
            };
    \addlegendentry{\textsf{ReFiT}}
    \addlegendentry{$\mathcal{O}\left( {\left| \mathcal{I} \right|} \right)$}
    \end{axis}
\end{tikzpicture}
\captionsetup{skip=0pt}
\caption{The computational complexity of \textsf{ReFiT(DiffRec)}, where the plots of the execution time versus $\left| \mathcal{U} \right|$ in Fig. \ref{fig:time}a and the execution time versus $\left| \mathcal{I} \right|$ in Fig. \ref{fig:time}b are shown.}
\label{fig:time}
\end{figure}
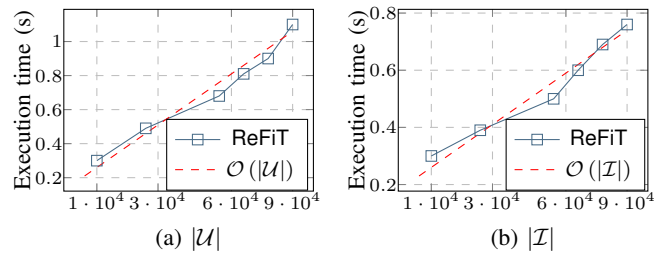

\section{Related Work}
\label{app:related_work}

Our proposed framework is related to three broader areas of research, namely 1) diffusion-based recommender systems, 2) recommender systems with RL, 3) and RL for diffusion models.
\subsection{Diffusion-Based Recommender Systems}

Diffusion models \cite{ho2020denoising, rombach2022high}, known for generating high-quality data, have been adapted to recommendation by iteratively recovering user–-item interactions via neural networks \cite{walker2022recommendation, wang2023diffusion}. To better capture collaborative signals, the high-order connectivity information was incorporated into diffusion-based recommender systems \cite{hou2024collaborative, zhu2024graph}. DDRM \cite{zhao2024denoising} improves recommender robustness via multi-step denoising of user/item embeddings. DiFashion \cite{xu2024diffusion} generated personalized fashion images for visually compatible outfit recommendations. HDRM \cite{yuan2025hyperbolic} leveraged hyperbolic geometry for topology-preserving diffusion, while S-Diff \cite{xia2025s} introduced a spectral-domain diffusion to recover user preferences. Additionally, a knowledge graph is used to further improve the model generalization of diffusion-based recommenders \cite{li2025mask}. Thanks to the success of diffusion-based recommenders for CF, attention has recently been paid to applying diffusion models to diverse recommendation tasks, such as sequential recommendations \cite{yang2023generate, wu2023diff4rec,li2023diffurec}, social recommendations \cite{li2024recdiff, liu2025score}, and POI recommendations \cite{zhao2020go, qin2023diffusion}.

\vspace{-0.3cm}
\subsection{Recommender Systems with RL}

RL-based recommender systems have recently gained attention for their interactive and autonomous learning capabilities, often modeling recommendation as an MDP and solving it using RL techniques such as value function, policy search, and actor-critic methods \cite{zou2020neural, wang2021reinforcement, zhang2019text, chen2019large, huang2021deep, liu2019exploiting, xiao2021general}. Value-based methods treat recommendation as decision-making, using user feedback to learn an exploration policy \cite{zou2020neural, wang2021reinforcement}. To handle large discrete action spaces, policy search methods directly optimize the policy by maximizing cumulative rewards from user feedback \cite{zhang2019text, chen2019large, huang2021deep}. Actor-critic algorithms combine value estimation and policy gradients to improve recommendation performance \cite{liu2019exploiting, xiao2021general}. However, RL-based recommenders often struggle with the exploration–exploitation trade-off, impacting learning efficiency and recommendation quality \cite{xie2021explore}.

\vspace{-0.3cm}
\subsection{RL for Diffusion Models}

Recent studies have explored the use of diffusion models in sequential decision-making, particularly in RL, which is categorized into policies \cite{WangHZ23, lu2023contrastive, kang2024efficient} and planners \cite{he2024diffusion, JannerDTL22, fan2023optimizing, black2023training, fan2024reinforcement, lee2023aligning}. As policies, diffusion models enhance expressiveness in Q-learning by sampling actions that yield high returns given the current state \cite{WangHZ23, lu2023contrastive, kang2024efficient}. As planners, diffusion models generate multi-step trajectories that optimize domain-specific rewards in tasks such as trajectory generation \cite{he2024diffusion, JannerDTL22} and image generation \cite{fan2023optimizing, black2023training, fan2024reinforcement, lee2023aligning}. These methods focus on image generation with external reward evaluation, making it difficult to extend such ideas to recommender systems without an external reward model.


\section{Conclusions and Outlook}

We explored an open yet fundamental problem of how to fine-tune diffusion models for recommendations with an aid of RL. To achieve this goal, we proposed \textsf{ReFiT}, a framework that effectively and efficiently integrates RL-aided fine-tuning into underlying diffusion-based recommender systems, guided by our collaborative signal-aware reward function. Extensive experiments on wide-ranging real-world benchmark datasets, we demonstrated that \textsf{ReFiT} (a) outperforms state-of-the-art CF methods across a wide spectrum of recommendation tasks, achieving improvements of up to 36.3\% in NDCG@20, (b) benefits significantly from its customized reward function, and (c) offers practical scalability, with linear computational complexity and substantial runtime savings empirically verified. Potential avenues of our future research include the design of a more robust reward function with the aid of large language models, leveraging their ability to understand user preferences and generate adaptive feedback.

\section*{Acknowledgments}
This research was supported by the National Research Foundation of Korea (NRF) funded by Korea Government (MSIT) under Grant RS-2021-NR059723 and Grant RS-2023-00220762.
\section*{Appendix A\\Algorithms}

\subsection*{I. Algorithm of ELBO-Based Fine-Tuning}
\label{app:ELBO_alg}
The algorithm of ELBO-based fine-tuning is summarized in Algorithm \ref{alg:ELBO}, which was used for experiments in Section IV-B3 of the main manuscript.

\begin{algorithm}
\caption{ELBO-based fine-tuning}\label{alg:ELBO}
\begin{algorithmic}[1]
\Require Pre-trained model $p_\theta$, where $\theta = \theta_0$, all users $\mathcal{U}$, number of iterations $Iters$, time step $T$, learning rate $l$. 
\While{$i < Iters$}
\State Sample a batch of users' interactions $U \subset \mathcal{U} $
\For{all user ${\bf u}_0  \in U $}
\State Compute ${\bf u}_{t} $ given ${\bf u}_0$ via $q\left( {u_t \left| {u_0 } \right.} \right) $ 
\State Compute $\mathcal{L}\left( \theta  \right) $ in (1)
\State Compute gradient $ \nabla _\theta = \nabla \mathcal{L}\left( \theta  \right) $
\State $\theta  \leftarrow \theta  - l \cdot \nabla _\theta  $

\EndFor
\State $i = i + 1 $
\EndWhile
\end{algorithmic}
\end{algorithm}
\vspace{-1em}

\subsection*{II. Inference of \textsf{ReFiT}}
\label{app:infer_alg}
The inference procedure of \textsf{ReFiT} is summarized in Algorithm \ref{alg:infer}.

\begin{algorithm}[h]
\caption{\textsf{Inference}}\label{alg:infer}
\begin{algorithmic}[1]
\Require Fine-trained model $p_\theta$, the interaction history ${\bf u}_0$ of user $u$, time step $T$. 
\State Sample noise $\boldsymbol{\epsilon} \sim \mathcal{N}(0, I)$.
\State Generate noisy interaction $\mathbf{u}_T$ given ${\bf u}_0$ and $\boldsymbol{\epsilon}$, and set $\hat{\mathbf{u}}_T = \mathbf{u}_T$ 
\For{$t = T$ to $1$}
\State Compute $\hat{\mathbf{u}}_{t-1} = \mu_\theta(\hat{\mathbf{u}}_t, t)$  
\State Update $\hat{\mathbf{u}}_t \leftarrow \hat{\mathbf{u}}_{t-1}$  
\EndFor

\end{algorithmic}
\end{algorithm}

\section*{Appendix B \\ Theoretical Discussion}
\label{app:proof}

In this section, we provide a more rigorous analysis of the loss of \textsf{ReFiT} in comparison with the ELBO-based loss. We consider a pre-trained model $p_\theta$ and a reward function $r\left( {{\bf u}_0 } \right) $. When we fine-tune the model $p_\theta$, we have
\begin{equation}
    \begin{array}{l}
     \mathbb{E}_{p_{\theta }} \left[ { - r\left( {{\bf u}_0 } \right)\log p_\theta  \left( {{\bf u}_0 } \right)} \right] \\ 
    \le \!\! \mathbb{E}_{p_{\theta}} \!\! \left[ {\! r \! \left( {\!{\bf u}_0 }\! \right) \!\! \sum\limits_{t = 2}^T \!\!{\mathbb{E}_{q\left( {\!{\bf u}_t \!\left| {{\bf u}_0 \!\! } \right.} \right)} \!\! \left[ {\textup{KL} \! \left( {\!q\left( {\!{\bf u}_{t - 1} \! \left| {{\bf u}_t ,{\bf u}_0 \! } \right.} \right) \! \left\| {p_\theta \!\! \left( {\!{\bf u}_{t - 1} \!\left| {{\bf u}_t \!} \right.} \!\right)} \right.} \right)} \right]} } \right] \! + \! C, \\ 
     \end{array}
     \label{theory1}
\end{equation}
where $\mathbb{E}_{p_\theta  } \left[  \cdot  \right] $ denote the expectation over all trajectories sampled from $p_\theta$; ${\mathbb{E}_{q\left( {{\bf u}_t \left| {{\bf u}_0 } \right.} \right)} \left[ {\textup{KL} \left( {q\left( {{\bf u}_{t - 1} \left| {{\bf u}_t ,{\bf u}_0 } \right.} \right)\left\| {p_\theta  \left( {{\bf u}_{t - 1} \left| {{\bf u}_t } \right.} \right)} \right.} \right)} \right]} $ is the loss at the $t$-th time step obtained using the ELBO in (1) of the main manuscript; and $C$ is a constant.

By similarly following the steps in \cite{ho2020denoising,fan2024reinforcement}, we have
\begin{equation}
\begin{array}{l}
  \;\;\;- \log p_\theta  \left( {{\bf u}_0 } \right)\\
  \le  - \log p_\theta  \left( {{\bf u}_0 } \right) + \textup{KL} \left( {q\left( {{\bf u}_{1:T} \left| {{\bf u}_0 } \right.} \right)\left\| {p_\theta  \left( {{\bf u}_{1:T} \left| {{\bf u}_0 } \right.} \right)} \right.} \right) \\ 
  = \mathbb{E}_{{\bf u}_{1:T}  \sim q\left( {{\bf u}_{1:T} \left| {{\bf u}_0 } \right.} \right)} \left[ {\log \frac{{q\left( {{\bf u}_{1:T} \left| {{\bf u}_0 } \right.} \right)}}{{p_\theta  \left( {{\bf u}_{0:T} } \right)}}} \right] \\ 
  \buildrel \Delta \over = \mathcal{L}_{\textup{ELBO}}. \\ 
 \end{array}
\label{eq:first}
\end{equation}
Then, it follows that
\begin{equation}
    \begin{array}{l}
 \;\;\;\mathcal{L}_{\textup{ELBO}} \\
 = \mathbb{E}_{{\bf u}_{1:T}  \sim q\left( {{\bf u}_{1:T} \left| {{\bf u}_0 } \right.} \right)} \left[ {\log \frac{{q\left( {{\bf u}_{1:T} \left| {{\bf u}_0 } \right.} \right)}}{{p_\theta  \left( {{\bf u}_{0:T} } \right)}}} \right] \\  
 = \sum\limits_{t = 2}^T \!{\mathbb{E}_{{\bf u}_{1:T}  \sim q\left( {{\bf u}_{1:T} \left| {{\bf u}_0 } \right.} \right)} \!\!\left[ {\textup{KL} \!\left( {q\left( {{\bf u}_{t - 1}\!\! \left| {{\bf u}_t ,\!{\bf u}_0 } \right.} \right)\!\left\| {p_\theta \! \left( {{\bf u}_{t - 1} \!\left| {{\bf u}_t } \right.} \right)} \right.} \right)} \right]} \!\!+ \!\!C, \\
 \end{array}
\end{equation}
where $C \!=\! \mathbb{E}_{{\bf u}_{1:T}} \!\left[ {\textup{KL} \!\left( {q\!\left( {{\bf u}_T\! \left| {{\bf u}_0 } \right.} \right)\!\left\| {p_\theta \! \left( {{\bf u}_T } \right)} \!\right.} \right)} \!-\! \log p_\theta \! \left( {{\bf u}_0 \left| {{\bf u}_1 } \right.} \right) \right ]$.

Multiplying both hand sides in (\ref{eq:first}) by $r\left( {{\bf u}_0 } \right)$ and taking the expectation over the trajectory sampled from $p_\theta$, we have
\begin{equation}
    \begin{array}{l}
     \mathbb{E}_{p_{\theta }} \left[ { - r\left( {{\bf u}_0 } \right)\log p_\theta  \left( {{\bf u}_0 } \right)} \right] \\ 
    \le \!\! \mathbb{E}_{p_{\theta}} \!\! \left[ {\! r \! \left( {\!{\bf u}_0 }\! \right) \!\! \sum\limits_{t = 2}^T \!\!{\mathbb{E}_{q\left( {\!{\bf u}_t \!\left| {{\bf u}_0 \!\! } \right.} \right)} \!\! \left[ {\textup{KL} \! \left( {\!q\left( {\!{\bf u}_{t - 1} \! \left| {{\bf u}_t ,{\bf u}_0 \! } \right.} \right) \! \left\| {p_\theta \!\! \left( {\!{\bf u}_{t - 1} \!\left| {{\bf u}_t \!} \right.} \!\right)} \right.} \right)} \right]} } \right] \! + \! C, \\ 
     \end{array}
     \label{theory1}
\end{equation}

Therefore, our proposed loss function enables direct optimization of the true log-likelihood, offering a more accurate and effective training objective than the standard ELBO-based approach.

\section*{Appendix C \\ Additional experimental evaluations}
\label{app:more}
\subsection*{I. Implementation Details}
\label{app:settings}
The pre-trained DiffRec \cite{wang2023diffusion} and CF-Diff \cite{hou2024collaborative} are downloaded from their well-trained sources. If pre-trained models on certain datasets are not provided, then we pre-train them using the parameter settings described in their original articles. We use the same data split as the pre-training stage on DiffRec and CF-Diff. The diffusion step $T$, which corresponds to the total number of steps in the Markov decision process (MDP), is set to 40, 10, and 10 on the ML-1M, Yelp2018, Anime datasets, respectively. The number of iterations is set to 500, with early stopping applied during fine-tuning. The number of sampled users in each interaction is given by $\left\{ {30, 50, 100, 200, 300} \right\} $.

\subsection*{II. Competitors}
\label{app:competitors}
To comprehensively demonstrate the superiority of \textsf{ReFiT}, we compare it against thirteen state-of-the-art recommendation methods, including eleven competitors for recommendations for standard collaborative filtering (CF), one competitor for sequential recommendations, one competitor for social recommendation, and one competitor for point-of-interest (POI) recommendation. We implemented all these methods using the parameter settings described in their original articles.

\begin{itemize}
    \item {\bf NICF \cite{zou2020neural}.} This method explores CF in an interactive setting, where recommender agents iteratively make recommendations and update user profiles based on interactive feedback.
    \item {\bf FCPO \cite{ge2021towards}.} This is a fairness-constrained RL recommendation algorithm that models the recommendation problem as a constrained MDP, addressing dynamically changing group labels for items. 
    \item {\bf NGCF \cite{wang2019neural}.} This is a new recommendation framework based on graph neural networks that explicitly encodes collaborative signals through high-order connectivities.
    \item {\bf LightGCN \cite{he2020lightgcn}.} This model simplifies graph convolutional networks (GCNs) for recommendation by focusing solely on neighborhood aggregation, eliminating feature transformation and nonlinear activation to enhance performance.
    \item {\bf SGL \cite{wu2021self}.} This method involves enhancing GCNs for recommendation by integrating self-supervised learning to improve accuracy and robustness.
    \item {\bf CFGAN \cite{chae2018cfgan}.} This a novel generative adversarial network (GAN)-based CF framework that employs vector-wise adversarial training to enhance recommendation accuracy.
    \item {\bf MultiDAE \cite{liang2018variational}.} This method employs denoising autoencoders to learn latent representations from corrupted user--item interactions for top-$N$ recommendations.
    \item {\bf RecVAE \cite{shenbin2020recvae}.} This is a new variational autoencoder for CF, introducing novel regularization and training techniques to improve recommendation performance.
    \item {\bf DiffRec \cite{wang2023diffusion}.} This method employs diffusion models to iteratively denoise user--item historical interactions, enhancing recommendation accuracy.
    \item {\bf CF-Diff \cite{hou2024collaborative}.} This method enhances recommendation accuracy in diffusion-based CF by leveraging high-order connectivity information.
    \item {\bf HDRM \cite{yuan2025hyperbolic}.} This is a hyperbolic-space diffusion recommender model that preserves user–-item graph topology by modeling anisotropic directional diffusion through radial and angular constraints in a hyperbolic latent space.
    \item {\bf DreamRec \cite{yang2023generate}.} This method shows a guided diffusion model that reshapes sequential recommendations by generating an oracle item from historical interactions.
    \item {\bf RecDiff \cite{li2024recdiff}.} This method uses a diffusion-based social denoising framework to enhance social recommendation by iteratively removing noise from user representations in the hidden space.
    \item {\bf Diff-POI \cite{qin2023diffusion}.} This method uses a diffusion-based model to enhance next POI recommendation by sampling the user's spatial preferences through a spatio-temporal graph encoder and a diffusion-based sampling strategy.

\end{itemize}

\subsection*{III. Further Experiments on Standard CF}
\label{app:F3}
In this subsection, we show additional experimental results and analyses on recommendations for standard CF in order to provide the full set of experiments for all datasets and models.

\newcommand{\pmstd}[2]{#1{\scriptsize$\;\pm\;$#2}}
\newcommand{\SD}[1]{{#1}} 

\begin{table*}[!t]\centering
\setlength\tabcolsep{6pt}
\small
\captionsetup{skip=0pt}
\caption{Mean $\pm$ standard deviation of \textsf{ReFiT} and diffusion baselines on R@10 and N@10.
All results are averaged over five runs with different random seeds. 
\textsf{ReFiT} achieves higher accuracy with markedly lower variance.}
\label{tab:std_refit}
\begin{tabular}{c|cc|cc|cc}
\toprule
& \multicolumn{2}{c|}{\textbf{ML-1M}} & \multicolumn{2}{c|}{\textbf{Yelp2018}} & \multicolumn{2}{c}{\textbf{Anime}}\\
\cmidrule(lr){2-3}\cmidrule(lr){4-5}\cmidrule(lr){6-7}
\textbf{Method} & \textbf{R@10} & \textbf{N@10} & \textbf{R@10} & \textbf{N@10} & \textbf{R@10} & \textbf{N@10}\\
\midrule
\textbf{DiffRec} & 0.1058 ± 0.0031 & 0.0901 ± 0.0027 & 0.0351 ± 0.0013 & 0.0414 ± 0.0011 & 0.2193 ± 0.0065 & 0.5196 ± 0.0058\\
\textbf{\textsf{ReFiT}(DiffRec)} & \underline{0.1083} ± 0.0016 & \underline{0.0918} ± 0.0013 & 0.0355 ± 0.0006 & 0.0417 ± 0.0005 & 0.2231 ± 0.0038 & 0.5211 ± 0.0035\\
\cmidrule{1-7}
\textbf{CF-Diff} & 0.1077 ± 0.0028 & 0.0912 ± 0.0025 & \underline{0.0363} ± 0.0012 & \underline{0.0425} ± 0.0010 & \underline{0.2263} ± 0.0058 & \underline{0.5271} ± 0.0053\\
\textbf{\textsf{ReFiT}(CF-Diff)} & \textbf{0.1103} ± 0.0015 & \textbf{0.0927} ± 0.0012 & \textbf{0.0367} ± 0.0005 & \textbf{0.0428} ± 0.0004 & \textbf{0.2283} ± 0.0036 & \textbf{0.5319} ± 0.0032\\
\bottomrule
\end{tabular}
\end{table*}

\subsubsection{Stability Analysis}
As shown in Table \ref{tab:std_refit}, \textsf{ReFiT} achieves consistently better performance (higher means) and more stable (lower standard deviations) training than the case of diffusion-based competing methods. 
This stability stems from the fact that fine-tuning in \textsf{ReFiT} begins from a well-initialized action space, allowing the reinforcement learning process to explore more effectively and converge reliably. 
In contrast, diffusion-based competitors such as DiffRec \cite{wang2023diffusion} and CF-Diff \cite{hou2024collaborative} exhibit relatively higher variability across runs, as they rely on random initialization and must learn optimal actions from scratch. 

\subsubsection{Impact of Our Reward Function}
\label{app:reward_diffrec}
Fig. \ref{fig:reward_diffrec} shows three different rewards according to iterations given the pre-trained DiffRec. It reveals a similar tendency to that of Fig. 5.

\begin{figure}[!t]
\centering
\begin{minipage}[b]{0.32\linewidth}
    \centering
    \includegraphics[width=\linewidth]{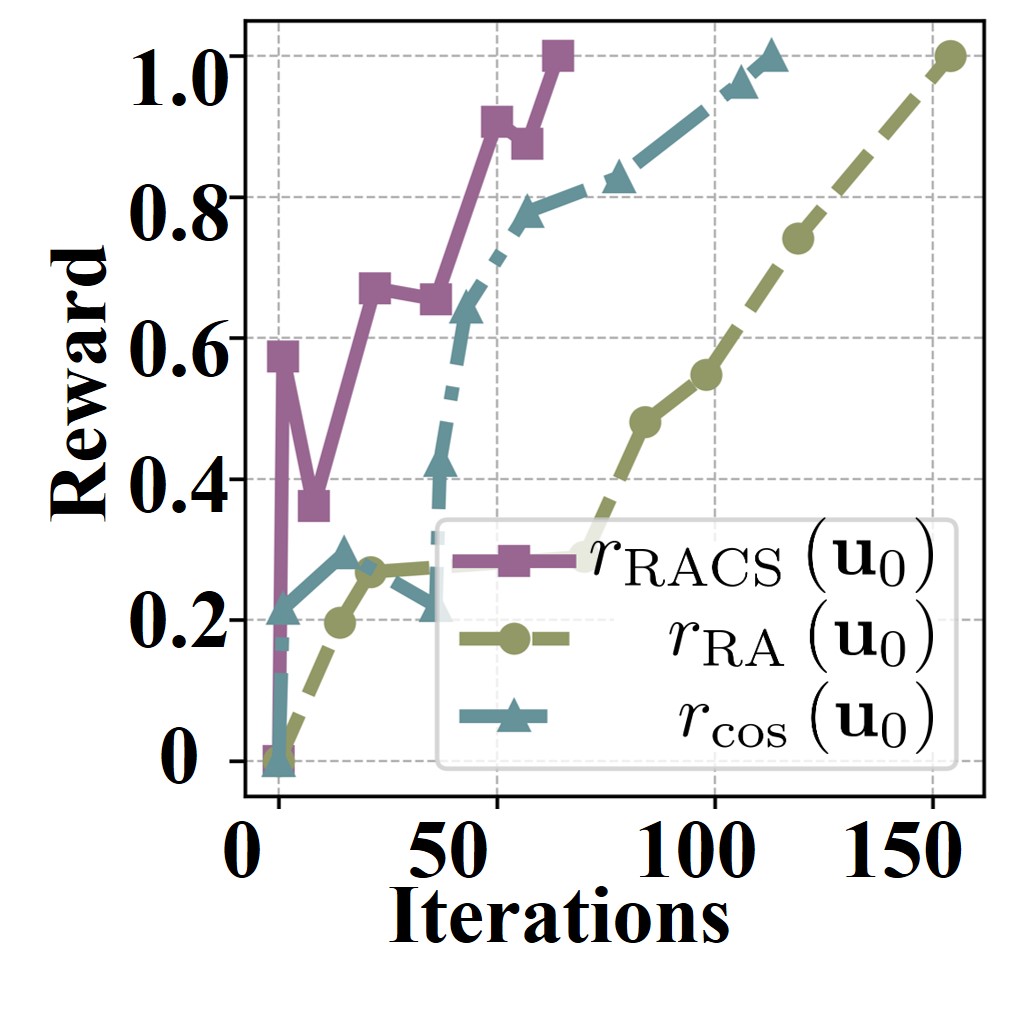}
    {\small\textbf{(a)} ML-1M}
\end{minipage}
\hfill
\begin{minipage}[b]{0.32\linewidth}
    \centering
    \includegraphics[width=\linewidth]{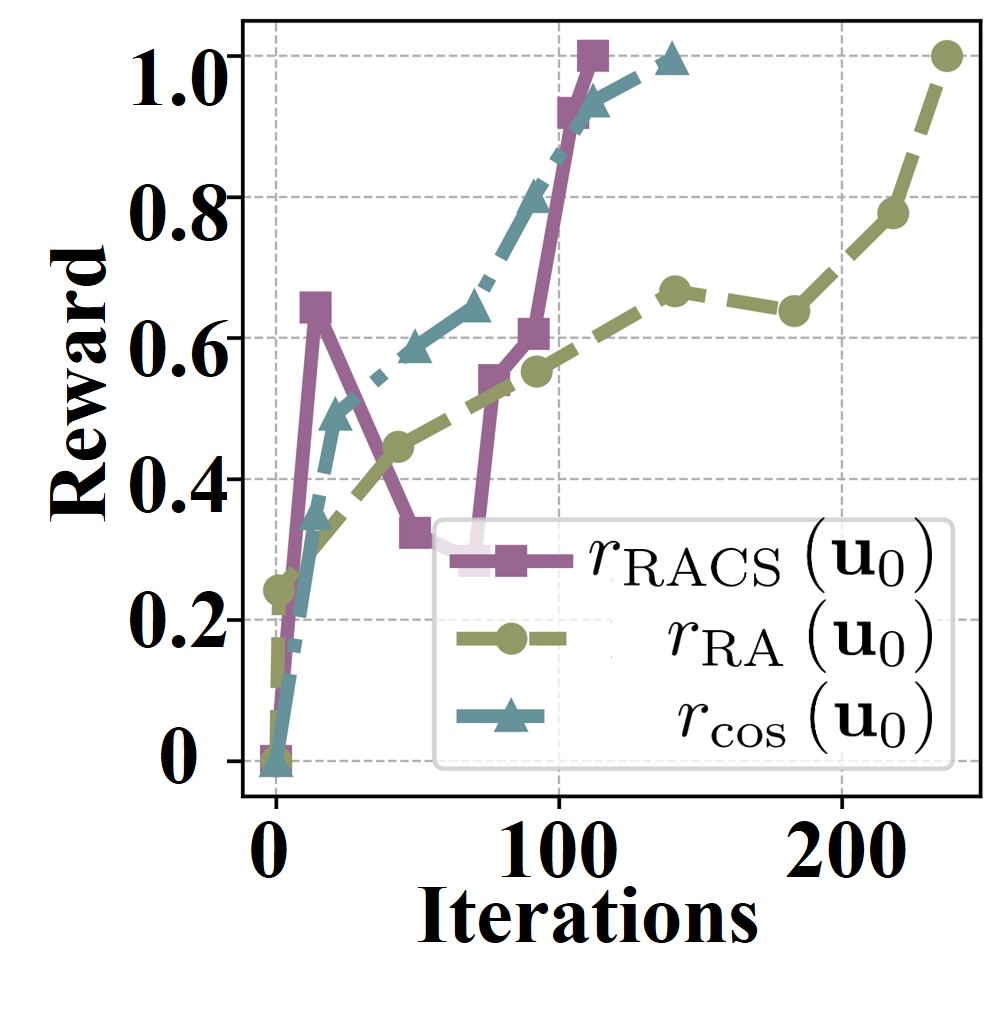}
    {\small\textbf{(b)} Yelp2018}
\end{minipage}
\hfill
\begin{minipage}[b]{0.32\linewidth}
    \centering
    \includegraphics[width=\linewidth]{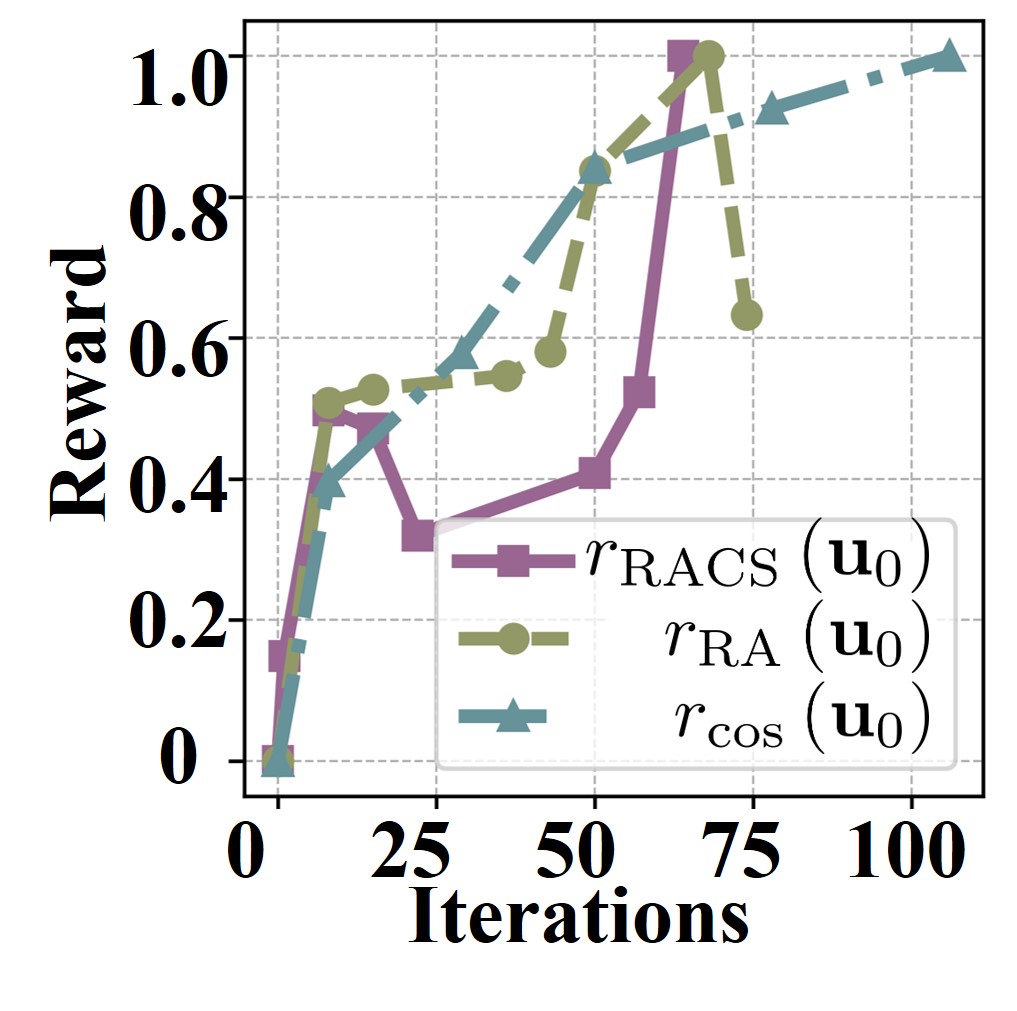}
    {\small\textbf{(c)} Anime}
\end{minipage}
\vspace{-3pt}
\caption{The behavior of different reward functions across iterations during fine-tuning with the pre-trained DiffRec.}
\label{fig:reward_diffrec}
\end{figure}

\begin{figure}[!t]
\centering
\begin{minipage}[b]{0.32\linewidth}
    \centering
    \includegraphics[width=\linewidth]{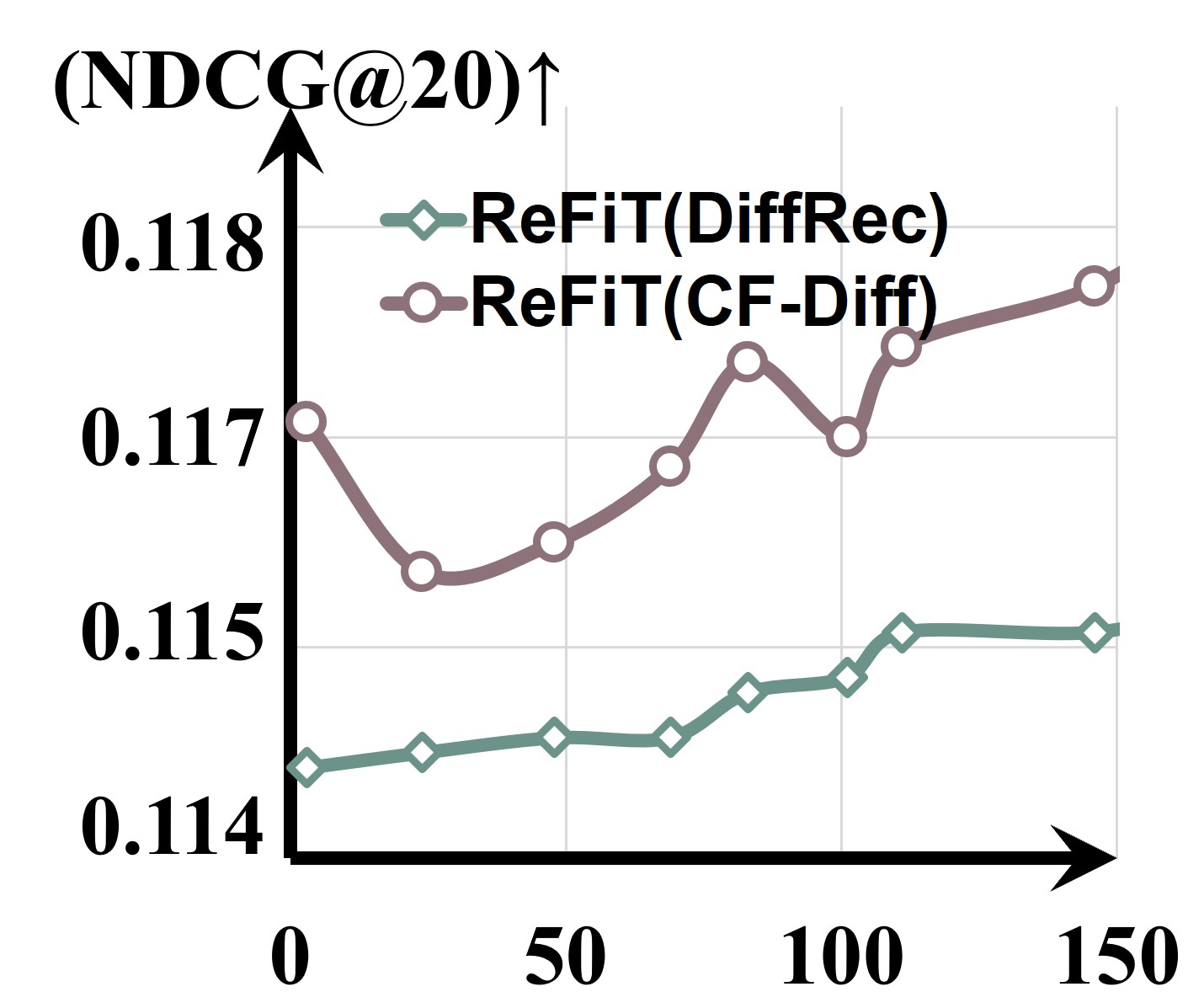}
    {\small\textbf{(a)} ML-1M}
\end{minipage}
\hfill
\begin{minipage}[b]{0.32\linewidth}
    \centering
    \includegraphics[width=\linewidth]{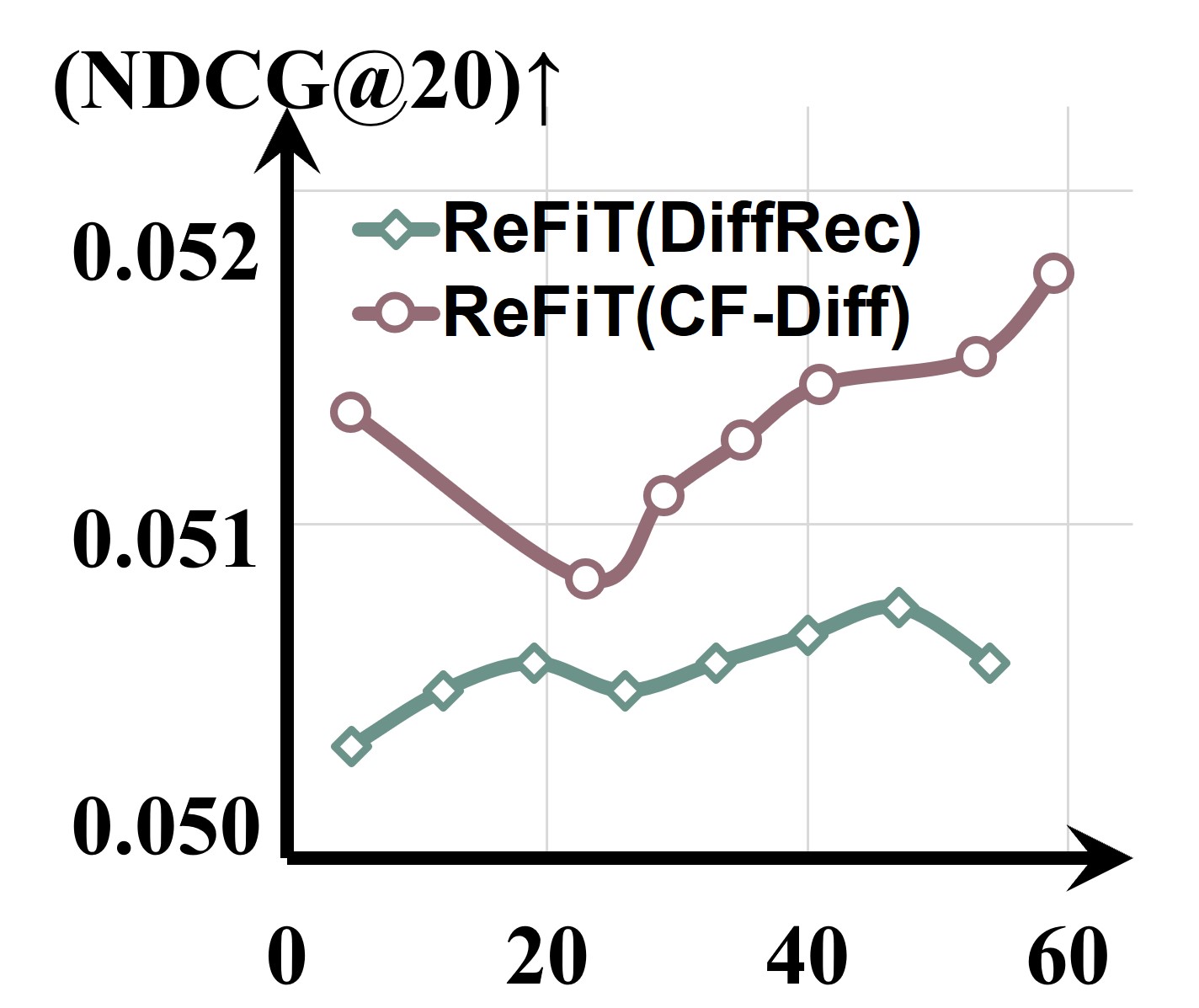}
    {\small\textbf{(b)} Yelp2018}
\end{minipage}
\hfill
\begin{minipage}[b]{0.32\linewidth}
    \centering
    \includegraphics[width=\linewidth]{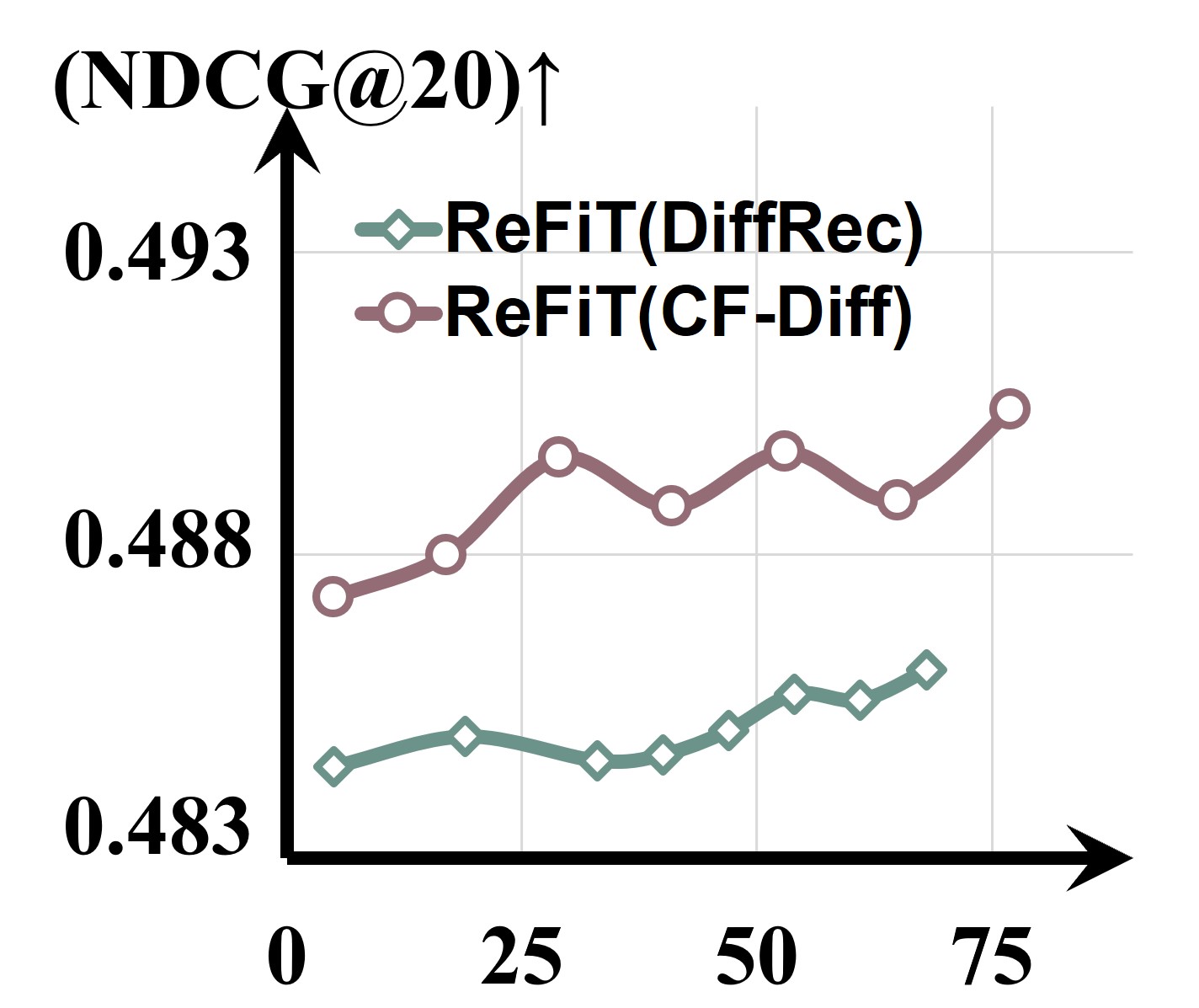}
    {\small\textbf{(c)} Anime}
\end{minipage}
\vspace{-3pt}
\caption{Recommendation accuracy in terms of NDCG@20 across fine-tuning iterations for \textsf{ReFiT(DiffRec)} and \textsf{ReFiT(CF-Diff)}.}
\label{fig:ndcg_vs_iters}
\end{figure}

\subsubsection{NDCG versus Iterations}
\label{app:reward_iter}

To better understand the fine-tuning behavior of \textsf{ReFiT}, we examine how the model’s recommendation accuracy evolves during training. While the experimental results in Section IV-B focus mainly on the final performance, examining the intermediate fine-tuning process helps reveal whether the optimization remains stable and whether \textsf{ReFiT} continues to learn effectively across iterations. Fig. \ref{fig:ndcg_vs_iters} illustrates the tendency of recommendation accuracy (NDCG@20) for both \textsf{ReFiT(DiffRec)} and \textsf{ReFiT(CF-Diff)} across fine-tuning iterations. As shown in Fig. \ref{fig:ndcg_vs_iters}, the performance consistently improves with the number of iterations, indicating that the model gradually adapts to the target recommendation task.

\bibliographystyle{IEEEtran}
\bibliography{bibfile}

\end{document}